\documentclass[11pt]{article}
\usepackage{xcolor}
\usepackage{geometry}   
\usepackage{setspace}
\usepackage[small,compact]{titlesec}
\titlespacing*\section{0pt}{12pt plus 4pt minus 2pt}{4pt plus 0pt minus 0pt}
\titlespacing*\subsection{0pt}{6pt plus 4pt minus 2pt}{2pt plus 4pt minus 2pt}
\titlespacing*\subsubsection{0pt}{12pt plus 4pt minus 2pt}{0pt plus 2pt minus 2pt}

\usepackage[title]{appendix}
\usepackage{natbib}
 \bibpunct[, ]{(}{)}{,}{a}{}{,}%
 \setcitestyle{citesep={;}}

\usepackage{subcaption}
\usepackage{amssymb}
\usepackage{amsmath}
\usepackage{mathtools}
\usepackage{multicol}
\usepackage{graphicx}
\usepackage{array}
\usepackage{bm}
\usepackage{soul}
\usepackage{dirtytalk}
\usepackage{algcompatible}
\makeatletter
\newenvironment{breakablealgorithm}
  {
   \begin{center}
     \refstepcounter{algorithm}
     \hrule height.8pt depth0pt \kern2pt
     \renewcommand{\caption}[2][\relax]{
       {\raggedright\textbf{\ALG@name~\thealgorithm} ##2\par}%
       \ifx\relax##1\relax 
         \addcontentsline{loa}{algorithm}{\protect\numberline{\thealgorithm}##2}%
       \else 
         \addcontentsline{loa}{algorithm}{\protect\numberline{\thealgorithm}##1}%
       \fi
       \kern2pt\hrule\kern2pt
     }
  }{
     \kern2pt\hrule\relax
   \end{center}
  }
\makeatother

\usepackage{algorithm}
\usepackage{caption}

\makeatletter
\newcommand\fs@boxedtopcap{\def\@fs@cfont{\bfseries}\let\@fs@capt\floatc@plain
	\def\@fs@pre{\setbox\@currbox\vbox{\hbadness10000
			\moveleft3.4pt\vbox{\advance\hsize by6.8pt
				\hrule \hbox to\hsize{\vrule\kern3pt
					\vbox{\kern3pt\box\@currbox\kern3pt}\kern3pt\vrule}\hrule}}}%
	\def\@fs@mid{\kern2pt}%
	\def\@fs@post{}\let\@fs@iftopcapt\iftrue}
\makeatother

\floatstyle{boxedtopcap}
\restylefloat{algorithm}

\usepackage{threeparttable}

\usepackage{enumitem}
\setlist[enumerate]{noitemsep, topsep=2pt}
\setlist[itemize]{noitemsep, topsep=2pt}

\usepackage{extarrows}

\usepackage{anysize}
\usepackage{latexsym}
\usepackage{amsthm}
\usepackage{epsfig}

\usepackage[T1]{fontenc}
\usepackage{mathrsfs}
\usepackage{float}
\usepackage{setspace}
\usepackage{color}
\definecolor{lawngreen}{RGB}{0,250,154}

\usepackage{titletoc}

\usepackage{color}
\usepackage{dsfont}
\usepackage{booktabs}
\definecolor{darkblue}{rgb}{0.0,0.0,0.5}
\definecolor{winered}{rgb}{0.5,0,0}
\usepackage[bookmarks,colorlinks,breaklinks]{hyperref}  
\hypersetup{linkcolor=winered,citecolor=winered,filecolor=winered,urlcolor=winered} 

\allowdisplaybreaks[1]

\graphicspath{{./Figures/}}

\setlist[enumerate]{noitemsep, topsep=0pt}

\usepackage{thmtools}

\theoremstyle{definition}
\newtheorem{theorem}{Theorem}
\newtheorem{definition}{Definition}
\newtheorem{lemma}{Lemma}

\newtheorem{remark}{Remark}

\newtheorem{corollary}{Corollary}

\newtheorem{proposition}{Proposition}

\def\E{\mathbb{E}}

\newcounter{note}[section]

\usepackage[nameinlink]{cleveref}
\crefname{assumption}{Assumption}{Assumptions}
\crefname{lemma}{Lemma}{Lemmas}
\crefname{theorem}{Theorem}{Theorems}
\crefname{corollary}{Corollary}{Corollaries}
\crefname{proposition}{Proposition}{Propositions}
\crefname{claim}{Claim}{Claims}
\crefname{procedure}{Procedure}{Procedures}
\crefname{algorithm}{Algorithm}{Algorithms}
\crefname{figure}{Figure}{Figures}
\crefname{remark}{Remark}{Remarks}
\crefname{section}{Section}{Sections}
\crefname{procedure}{Procedure}{Procedures}
\crefname{table}{Table}{Tables}
\crefname{appendix}{Appendix}{Appendices}
\crefname{example}{Example}{Examples}
\crefname{equation}{}{}



\usepackage{xr}
\usepackage{authblk}
 
\begin{document}

\title{Multi-LLM Query Optimization}

\author{Arlen Dean} 
\affil{Washington University in St. Louis, Olin Business School, arlen@wustl.edu} 
\author{Zijin Zhang}
\affil{Boston College, Carroll School of Management, zijin.zhang@bc.edu}
\author{Stefanus Jasin}
\affil{University of Michigan, Stephen M. Ross School of Business, sjasin@umich.edu}
\author{Yuqing Liu}
\affil{University of Michigan, Industrial and Operations Engineering, lyuqing@umich.edu}
\date{}



\maketitle
\vspace{-0.3in}

\begin{abstract}
Deploying multiple large language models (LLMs) in parallel to classify  an unknown ground-truth label is a common practice, yet the problem of optimally allocating queries across heterogeneous models remains poorly understood. 
In this paper, we formulate a robust, offline query-planning problem that minimizes total query cost subject to statewise error constraints which guarantee reliability for every possible ground-truth label. We first establish that this problem is NP-hard via a reduction from 
the minimum-weight set cover problem. To overcome this intractability, we develop a surrogate by combining a union bound decomposition of the multi-class error into pairwise comparisons with Chernoff-type concentration bounds. The resulting surrogate admits a closed-form, multiplicatively separable expression in the query counts and is guaranteed to be feasibility-preserving. We further show that the surrogate is asymptotically tight at the optimization level: the ratio of surrogate-optimal cost to true 
optimal cost converges to one as error tolerances shrink, with an explicit rate of $\mathcal{O}\left(\log\log(1/\alpha_{\min}) / \log(1/\alpha_{\min})\right)$. 
Finally, we design an asymptotic fully polynomial-time approximation scheme (AFPTAS) that returns a surrogate-feasible query plan within a  $(1+\varepsilon)$ factor of the surrogate optimum.
\end{abstract}
\textit{Keywords: large language models, query optimization, polynomial-time approximation}

\section{Introduction}
Large language models (LLMs) are rapidly transforming how organizations operate. Rather than relying on a single model, many deployments use a collection of LLMs and query them on the \emph{same} task to  aggregate their responses \citep{ai2025beyond, wang2024mixture}.  Combining multiple predictors is well established in classical machine learning, where ensemble methods such as stacking can improve accuracy beyond any single learner \citep{wolpert1992stacked,breiman1996stacked}. 
However, the case for aggregation is especially compelling for LLMs because different models can disagree substantially due to differences in training data and architecture, and their outputs are inherently stochastic.
Pooling across models and across repeated queries is therefore a natural mechanism for reducing this variability and improving reliability \citep{kuncheva2003measures}.

One prominent application of this multi-model approach is prediction, where a system consults several LLMs to infer an unknown ground-truth state from a finite set of outcomes. Such a use appears across many domains. In healthcare, a hospital system might run the same case through several diagnostic models and then aggregate their predicted diagnosis from a fixed menu of conditions \citep{yuan2026mixed, more2026theramind}. In streaming or online marketplaces, platforms may combine outputs from multiple LLMs to classify customer intent into a finite set of content or product categories \citep{fang2024llm, boateng2026agentic}. In legal services, document-review pipelines may route the same filing through several models and aggregate their labels for text classification \citep{huffman2025exploiting}. In these settings, using multiple LLMs in parallel on the same unknown ground truth can yield substantial improvements over any single model.

Realizing these benefits in practice, however, is far from straightforward. 
Aggregating responses from multiple LLMs requires the decision maker to first determine how often to consult each model, based on their understanding of how each model performs across different categories of a given task.
Yet each query consumes compute, adds latency, and incurs monetary cost through API fees or related charges \citep{shekhar2024towards}.
Querying all available models as many times as possible is therefore inefficient and wasteful, while querying too few may leave certain ground-truth labels poorly covered, resulting in unacceptable misclassification for those categories. Navigating between these extremes is difficult because models are heterogeneous in their discriminative power. A model that reliably distinguishes one pair of categories may be nearly uninformative for another, and repeated queries to the same model reduce uncertainty but with diminishing returns. Existing studies may offer partial guidance on individual model behavior, but they do not reveal how to distribute a limited budget across models whose strengths, weaknesses, and costs all interact. What is needed is an efficient allocation of queries across models, determined in advance, so that the aggregated responses are sufficient to reliably classify the unknown label. In practice, this allocation is largely handled through ad-hoc heuristics and trial-and-error. \textbf{This gap  motivates us to ask the following research questions:} \vspace{2mm}

\begin{enumerate}
\item[\emph{(i)}] \emph{How should a decision maker allocate queries across LLMs to minimize cost while guaranteeing a target level of  accuracy across all possible ground-truth labels?}\vspace{2mm}

\item[\emph{(ii)}] \emph{How hard is this optimization problem, and how can we solve it tractably with provable guarantees?} 
\end{enumerate}

\subsection{Our Main Results and Contributions}
\label{ss:contribution}
Our main results and contributions are as follows:

\textbf{A principled framework for offline multi-LLM query planning.} We formulate a robust optimization model for the offline query-planning problem that arises whenever an organization consults multiple LLMs to classify an unknown state. The formulation captures three features that are central to practical deployments but have not been jointly addressed in the literature, namely (i) heterogeneous per-query costs across models, (ii) model-specific and state-dependent discriminative power, and (iii) statewise error constraints that require reliability for every possible ground-truth label, not just on average. The framework applies broadly to any firm that employs multiple LLMs for classification or prediction tasks and provides a rigorous foundation for the upstream planning decision of allocating a query budget across a heterogeneous collection of models.

\textbf{NP-hardness and a tractable Chernoff surrogate.} We establish that the query design problem is NP-hard via a reduction from the minimum-weight set cover problem (\cref{thm:nphard}). The reduction exploits the fact that ensuring correct classification of every label forces the decision maker to select a collection of models whose combined discriminative reach spans all label pairs, at minimum total cost. To overcome this intractability, we develop a tractable surrogate that replaces each intractable statewise error constraint with an efficiently computable upper bound (\cref{thm:statewise_chernoff}). The construction proceeds in two steps. A union bound decomposes the multi-class probability error into pairwise label comparisons (\cref{lem:union_bound}), and a Chernoff exponential bound controls each pairwise term (\cref{prop:pairwise_chernoff}). 
The resulting expression of the upper bound factorizes across models and query counts, which makes constraint evaluation efficient. Moreover, feasibility under the surrogate automatically implies feasibility under the original constraints.

\textbf{Asymptotically Optimal Polynomial-time Approximation.} We show that replacing the exact error constraints with the Chernoff surrogate does not significantly inflate the optimal query cost. Specifically, we prove that the ratio of surrogate-optimal cost to true optimal cost converges to one as the error tolerances shrink, with the excess vanishing at an explicit rate that is logarithmically smaller than the dominant cost scale (\cref{thm:optight}). The gap arises because the Chernoff proxy captures the correct exponential decay rate of each pairwise error probability but introduces a polynomial prefactor. Closing this prefactor requires only a negligible amount of additional query effort relative to the total budget any feasible plan must already spend. This result establishes that, in the high-reliability regime, solving the tractable surrogate yields essentially the same minimum cost as solving the intractable exact problem. The surrogate is therefore not merely a convenient analytical relaxation. It preserves the first-order cost structure that governs how query effort should be distributed across models. Building on this foundation, we design an asymptotic fully polynomial-time approximation scheme (AFPTAS) that returns a surrogate-feasible query plan whose cost is within a factor of $1+\varepsilon$ of the surrogate optimum for any desired accuracy $\varepsilon > 0$ (\cref{thm:afptas}).

\subsection{Notations}
Throughout the paper, we use $\|\cdot\|_p$ to denote the $p$-norm of a vector in $\mathbb R^d$, and $\|x\|_{\mathbf A} := \sqrt{x^\top \mathbf A x}$ to denote the weighted 2-norm of the vector $x \in \mathbb R^d$ where $\mathbf A \in \mathbb{R}^{d \times d}$ is a positive definite matrix.  We use the big-$O$ notation where, by definition, $f(x) = O\left(g(x)\right)$ for positive real-valued functions $f$ and $g$ if there exists an $r \in \mathbb R_+$ such that $f(x) < r g(x)$. Similarly, if $f(x) = \Omega \left(g(x)\right)$, then $f(x)>rg(x)$. When $f(x) = O(g(x))$ and $f(x) = \Omega(g(x))$, it is represented by $f(x) = \Theta (g(x))$.
For any numbers \( a, b \in \mathbb{R} \), we use the notation: $a \vee b := \max\{a, b\}$ and $ a \wedge b := \min\{a, b\}$, to denote the maximum and minimum of $a$ and $b$, respectively. $(a)^+$ denotes $\max\{a,0\}$.


\section{Literature Review} \label{s:literature}
\paragraph{Multi-LLM systems.}
A growing body of work studies how to combine or coordinate multiple LLMs to attain decision quality beyond that of any single model. The simplest aggregation strategy is majority voting, introduced in this setting through self-consistency decoding \citep{wang2022selfconsistency}, where diverse chain-of-thought paths are sampled and the most frequent answer is selected. More sophisticated ensembling methods include pairwise ranking \citep{jiang2023llmblender} and approaches that exploit higher-order relationships across model outputs \citep{ai2025beyond}. For a recent survey, see \citet{chen2025harnessing}. A related paradigm is multi-agent debate, in which several models iteratively refine their answers through interaction \citep{du2024improving, liang2023encouraging}. \citet{wang2024mixture} study this idea through a mixture-of-agents architecture, while \citet{li2024moreagents} examine how performance scales with the number of agents. A separate line of work considers per-query model selection rather than output aggregation. LLM routing methods assign each query to a single model to balance cost and quality \citep{ong2025routellm, ding2024hybrid}, while LLM selects query models sequentially in order of increasing cost and stop once a confidence threshold is met \citep{chen2023frugalgpt, jitkrittum2024confidence, gupta2024cascades, dekoninck2025unified}. 
Across these studies, the focus is on how to aggregate or route outputs after they are observed. By contrast, our work addresses an upstream design problem: how many times to query each model before any outputs are collected, subject to cost and accuracy constraints.

\paragraph{Crowdsourcing and noisy label aggregation.}
The problem of aggregating responses from multiple noisy sources to infer a latent ground truth has a long history in statistics and crowdsourcing. The foundational model of \citet{dawid1979maximum} estimates annotator error rates via EM when the true label is unobserved. Our setting shares its core structure, with multiple sources that have heterogeneous accuracies, a latent class, and inference based on posterior maximization. Subsequent work extends this framework through richer classification models and alternative inference methods \citep{whitehill2009whose, liu2012variational}. On the algorithmic side, \citet{zhang2014spectral} develop spectral methods that improve on naive voting, while \citet{raykar2010learning} study cost-quality tradeoffs in learning from crowd labels. Information-theoretic limits are also well understood. \citet{gao2016exact} establish matching upper and lower bounds with an exact exponent characterized by average Chernoff information across workers, which parallels the role of pairwise Chernoff bounds in our surrogate construction. Moving from inference to \emph{budget allocation}, \citet{karger2011iterative} and \citet{karger2014budget} study budget-optimal non-adaptive task assignment in crowdsourcing. 
Their formulation is the closest structural antecedent to our problem. More broadly, the crowdsourcing literature typically assumes homogeneous query costs and considers assignment across many items. By contrast, we optimize integer query counts for a single classification instance with heterogeneous per-model costs, which results in a challenging combinatorial optimization problem.

\paragraph{Active Learning.}
Our query-planning problem is also related to active learning, which studies how to select informative queries under budget constraints \citep{settles2012active}. Much of this literature is adaptive, meaning that the learner observes each outcome before deciding what to query next. Classical examples include uncertainty sampling \citep{lewis1994sequential}, query-by-committee \citep{seung1992query}, and variance-based statistical approaches \citep{cohn1996active}. These settings differ from ours because we require the decision maker to commit to a complete query plan before any outputs are observed. In this respect, our problem is closer to batch-mode active learning, which selects an entire set of queries in a single round \citep{hoi2006batch, chen2013nearoptimal, zhang2024more}, and to cost-sensitive active learning with multiple imperfect oracles, where query decisions must account for heterogeneous costs and reliabilities \citep{donmez2008proactive, yan2011active}. More broadly, this line of work connects to the classical literature on the value of information in hypothesis testing and sequential design \citep{chernoff1952measure, chernoff1959sequential}. Within this broader perspective, our formulation belongs to a non-adaptive class of information acquisition problems such as  controlled sensing (e.g., \citealt{nitinawarat2013controlled}), where non-adaptive observation policies can be asymptotically optimal for multihypothesis testing under fixed sample budgets .

\section{Model and Preliminary Analysis} \label{s:model}

We consider a finite classification task in which a decision maker queries multiple large language models (LLMs) and aggregates their responses to predict an unknown ground-truth label.
The unknown ground-truth label is a random variable, denoted by $Y \in \mathcal{Y}:= \{1, 2, \ldots, L\}$, and the decision maker holds a strictly positive prior $\pi(y) := \Pr(Y = y) > 0$ for every label $y \in \mathcal{Y}$. 
To learn about the true label $Y$, the decision maker has access to $K$ different LLMs, indexed by $m \in \{1, \ldots, K\}$. A single query to model $m$ produces a random response $X_{m,t} \in \mathcal{X}_m$ drawn from the conditional distribution $p_m(x \mid y) := \Pr(X_{m,t} = x \mid Y = y)$, where $\mathcal{X}_m$ represents the output alphabet of model $m$ and $t$ indexes repeated queries. 
For every model $m\in\{1,\ldots,K\}$, every label $y\in\mathcal{Y}$, and every symbol $x\in\mathcal{X}_m$, we have $p_m(x\mid y)>0$.
Conditional on $Y = y$, the outputs of each model are independent and identically distributed across repetitions, and outputs are mutually
independent across models.
The joint likelihood of all observed outputs therefore factorizes as
\begin{equation}\label{eq:joint_likelihood}
  \Pr \bigl(\{X_{m,t} = x_{m,t}\}_{m,t} \mid Y = y\bigr)=
  \prod_{m=1}^{K} \prod_{t=1}^{r_m} p_m(x_{m,t} \mid y),
\end{equation}
where $r_m$ denotes the number of times that model $m$ is queried.

Our goal is to design a \emph{query plan}, specified by a query-count vector $r = (r_1, \ldots, r_K) \in \mathbb{Z}_{\ge 0}^K$.
The query plan is \emph{non-adaptive}: the decision maker commits to a fixed number of queries per model before observing any outputs.
Each query to model $m$ incurs a known cost $c_m > 0$, and the total cost by a plan is denoted by  $C(r) := \sum_{m=1}^{K} c_m\, r_m$. 
Once $r$ is fixed, the decision maker observes the dataset $X(r) := \{X_{m,t} : 1 \le m \le K, 1 \le t \le r_m\}$. 
Given the realized observations $X(r)=x(r)$, the decision maker updates their belief about the unknown ground-truth label $Y$ using Bayes' rule combined with the joint likelihood in~\eqref{eq:joint_likelihood}. The posterior distribution of $Y$ thus follows:
\begin{equation}\label{eq:posterior}
  \Pr\bigl(Y = y \mid X(r) = x(r)\bigr)
  =
  \frac{
    \pi(y)\;\prod_{m=1}^{K}\prod_{t=1}^{r_m} p_m(x_{m,t} \mid y)
  }{
    \sum_{y' \in \mathcal{Y}}
    \pi(y') \prod_{m=1}^{K}\prod_{t=1}^{r_m} p_m(x_{m,t} \mid y')
  }.
\end{equation}
The decision maker predicts label $\hat{Y}$ using the maximum-a-posteriori (MAP) rule, which selects the label with the largest posterior probability. Since the denominator in~\eqref{eq:posterior} does not depend on $y$, maximizing the posterior probability is equivalent to solving

\begin{equation}\label{eq:map_estimator}
  \widehat{Y}_r\bigl(x(r)\bigr)
  \in
  \operatorname*{arg\,max}_{y \in \mathcal{Y}}
  \Bigl\{\pi(y)\prod_{m=1}^{K}\prod_{t=1}^{r_m} p_m(x_{m,t} \mid y)
  \Bigr\},
\end{equation}
which we refer to as the \emph{MAP estimator}.
We evaluate the query performance through a \emph{statewise} error probability 
$$P_e(y; r) := \Pr(\widehat{Y}_r(X(r)) \neq y \mid Y = y),$$ which measures the misclassification rate conditional on the true label being $y$.
This criterion is \emph{robust}: it requires the error probability to be small for every possible ground-truth label, not just on average under a prior. Given a vector of target tolerances $\alpha = (\alpha_y)_{y \in \mathcal{Y}} \in (0,1)^L$, the objective is to find the minimum-cost query plan satisfying every statewise constraint simultaneously:
\begin{equation}\label{eq:opt_problem}
  \begin{aligned}
    \min_{r \in \mathbb{Z}_{\ge 0}^K} \quad
      & C(r) = \sum_{m=1}^{K} c_m r_m \\
    \text{s.t.} \quad
      & P_e(y; r) \le \alpha_y,
        \qquad \forall y \in \mathcal{Y}.
  \end{aligned}
\end{equation}

Problem~\eqref{eq:opt_problem} is an integer program whose constraints couple the query counts through the combinatorial structure of the MAP error.  To ensure that the problem is feasible for all sufficiently
small tolerance vectors, we require that every pair of distinct labels be
distinguishable by at least one model: for every label $y \neq y'$, there exists model $m$ with $p_m(\cdot \mid y) \neq p_m(\cdot \mid y')$. 
The remainder of this paper aims to develop an approach for designing cost-optimal query plans in this offline robust setting.


\subsection{Hardness Result}
Directly optimizing problem~\eqref{eq:opt_problem} is computationally intractable for two main reasons. First, the exact evaluation of $P_e(y; r)$ requires summing over all possible observation sequences, resulting in a combinatorial sum that grows exponentially in $\sum_{i=1}^m r_i$ and is therefore infeasible to compute in general. Second, evaluating $P_e(y; r)$ also involves solving an inner nonlinear optimization problem to determine the estimator $\hat{Y}_r$ within constraints. Taken together, these challenges make the problem difficult to solve exactly. We formally establish this computational difficulty by showing that the problem is NP-hard.
For the remainder of the paper, all proofs are deferred to the Appendix.

\begin{theorem}[NP-Hardness]\label{thm:nphard}
The query design problem~\eqref{eq:opt_problem} is
\textup{NP}-hard.
\end{theorem}

\cref{thm:nphard} shows that our optimization problem of solving the minimum-cost query plan is NP-hard. This can be proved by establishing a polynomial-time reduction from the minimum-weight set cover problem. The intuition is that each error constraint imposes a covering requirement: we can construct an instance such that for each potential ground-truth label to be correctly classified, at least one model that is capable of distinguishing it must be queried. Different models distinguish different subsets of labels, so the decision maker faces a combinatorial selection problem in choosing the cheapest collection of models whose discriminative capabilities jointly cover all labels. 
This is precisely the structure of the set cover problem, which is NP-hard. 
Therefore, no polynomial-time algorithm can solve the query design problem \eqref{eq:opt_problem} exactly for all instances unless P = NP. 
To address this computational challenge, we next introduce a surrogate design approach.


\section{Query Design: Surrogate Problem}
\label{s:algorithm}
A central challenge in multi-model query design is controlling the statewise classification error without explicitly computing posterior probabilities.
As shown in \cref{thm:nphard}, directly evaluating the statewise error probability $P_e(y;r)$ requires enumerating all possible observation sequences.
When the number of queries grows, the number of such sequences increases exponentially, which makes exact evaluation computationally infeasible.

To overcome this difficulty, this section develops an analytically tractable upper bound on the statewise MAP error probabilities $P_e(y;r)$ by combining a union bound over pairwise comparisons with Chernoff bounding techniques.
The resulting bound admits a closed-form expression and exhibits a key structural property: it decomposes multiplicatively across models and query counts.
This separability allows us to transform the original query design problem into a more tractable optimization problem.

The key idea of our approach is to replace each intractable constraint $P_e(y;r)\le \alpha_y$ with a tractable surrogate constraint $\overline{P}_e(y;r)\le \alpha_y$, where $\overline{P}_e(y;r)$ is constructed as a valid upper bound on the true error probability. Therefore, any query plan that satisfies the surrogate constraints automatically satisfies the original ones.
At the same time, the surrogate bound admits a simple closed-form expression in terms of the query plan $r$.
This leads to the following optimization problem
\begin{equation}\label{eq:surrogate_problem}
  \begin{aligned}
    \min_{r \in \mathbb{Z}_{\ge 0}^K} \quad
      & C(r) = \sum_{m=1}^{K} c_m r_m \\
    \text{s.t.} \quad
      & \overline{P}_e(y; r) \le \alpha_y,
        \qquad \forall y \in \mathcal{Y}.
  \end{aligned}
\end{equation}

Below, we construct $\overline{P}_e(y;r)$ through a sequence of steps. We first decompose the MAP error event into pairwise comparisons between the true label and competing labels. We then bound each pairwise comparison probability using a Chernoff bound.
Finally, combining these components yields a tractable upper bound that retains the essential statistical structure of the original problem.

\subsection{Upper Bound Construction}
\label{subsec:upper_bound}
The MAP rule selects the label whose posterior probability is largest.
Equivalently, it selects the label that maximizes the posterior log-likelihood
\[
  \widehat{Y}_r\bigl(x(r)\bigr)
  \in
  \operatorname*{arg\,max}_{y \in \mathcal{Y}}~ \log \pi(y) + 
  \sum_{m=1}^{K}\sum_{t=1}^{r_m}
  \log p_m(x_{m,t}\mid y).
\]
An error occurs when some competing label $y'$ achieves a posterior log-likelihood at least as large as that of the true label $y$.
This observation motivates analyzing the classification error through a collection of pairwise comparisons between the true label and each competitor.

\begin{definition} [Log-likelihood difference]
For any query plan $r\in\mathbb{Z}_{\ge 0}^K$ and distinct labels $y,y'\in\mathcal{Y}$, the \emph{log-likelihood difference} is a random variable
\begin{equation}
\label{eq:delta_expanded}
\Delta_{y,y'}(r):=
\log \left(\frac{\pi(y')}{\pi(y)}\right)
+
\sum_{m=1}^K\sum_{t=1}^{r_m}
\log \left(\frac{p_m(X_{m,t}\mid y')}{p_m(X_{m,t}\mid y)}\right).
\end{equation}
\end{definition}

The quantity $\Delta_{y,y'}(r)$ measures the advantage of the competing label $y'$ over the true label $y$ in terms of posterior log-likelihood.
It consists of two components.
The first term reflects the prior preference between the two labels, while the second term aggregates the evidence provided by the observed data.
Each observation contributes a log-likelihood ratio comparing how likely that observation is under the two competing hypotheses.
If the data is informative, this log-likelihood difference tends to favor the true label, causing $\Delta_{y,y'}(r)$ to drift downward as more evidence is collected; otherwise, classification errors occur when the data incorrectly favors a competing label.
Conditional on $Y=y$, the event $\{\Delta_{y,y'}(r)\ge 0\}$ indicates that a competing label is favored.
Hence, the error event $\{\hat{Y}_r(X(r)) \neq y\}$ is contained in the union of events $\{\Delta_{y,y'}(r)\ge 0\}$ over all competitors $y'\ne y$.
This observation allows us to express the error event as a union of pairwise comparison events. The following lemma formalizes this reduction.

\begin{lemma}[Union bound reduction to pairwise comparisons]
\label{lem:union_bound}
For any query plan $r\in\mathbb{Z}_{\ge 0}^K$ and label $y\in\mathcal{Y}$,
\begin{equation}
\label{eq:union_bound}
P_e(y;r)
\le
\sum_{y'\in\mathcal{Y}\setminus\{y\}}
\Pr \big(\Delta_{y,y'}(r)\ge 0 \,\big|\, Y=y\big).
\end{equation}
\end{lemma}

This lemma shows that the probability of misclassifying label $y$ can be controlled by bounding the probability that any competing label outperforms $y$ in terms of posterior log-likelihood.
Conceptually, this step replaces a multi-class classification problem into a collection of binary tests.
Each term in the sum corresponds to the probability that a particular competitor $y'$ beats the true label.
This union bound provides a useful analytical simplification. Instead of analyzing the full multi-class error structure, we can focus on understanding the behavior of individual pairwise comparisons.
However, the resulting probabilities are still difficult to compute exactly, because they depend on the enumeration of all possible observation sets. 
To obtain a tractable expression, we further bound these probabilities using a Chernoff parameter defined below.

\begin{definition}[Chernoff affinity factor]
\label{def:chernoff_factor}
For any two distinct labels $y,y'\in\mathcal{Y}$, model
$m\in\{1,\dots,K\}$, and tilting parameter $s\in[0,1]$, we define the
\emph{Chernoff affinity factor}
\begin{equation}
\label{eq:M_def}
M_{m}^{(y,y')}(s)
\;:=\;
\sum_{x\in\mathcal{X}_m}
  p_m(x\mid y)^{\,1-s}\;p_m(x\mid y')^{\,s}.
\end{equation}
\end{definition}

The Chernoff affinity factor measures the statistical overlap between the distributions induced by labels $y$ and $y'$ under model $m$.
When the two distributions are very similar, the value of
$M_m^{(y,y')}(s)$ is close to one, reflecting the fact that observations from model $m$ provide little information for distinguishing the two labels.
On the other hand, when the distributions are well separated, $M_m^{(y,y')}(s)$ becomes much smaller than one, indicating that the model is highly informative for this pairwise comparison.

From an operational perspective, this quantity captures the discriminative power of each model.
It tells us how rapidly evidence accumulates in favor of the correct label when we collect observations from model $m$.
This interpretation will later play a key role in determining how query effort should be allocated across models.
Using the Chernoff affinity factor, we can further bound the pairwise probability error in \cref{prop:pairwise_chernoff}.


\begin{proposition}[Pairwise Chernoff bound]
\label{prop:pairwise_chernoff}
For any query vector $r\in\mathbb{Z}_{\ge 0}^K$ and distinct labels $y,y'\in\mathcal{Y}$,
\begin{equation}
\label{eq:pairwise_chernoff_opt}
\Pr\!\big(\Delta_{y,y'}(r)\ge 0 \,\big|\, Y=y\big)
\le
\min_{s\in[0,1]}
\left(\frac{\pi(y')}{\pi(y)}\right)^{s}
\prod_{m=1}^K \big(M_{m}^{(y,y')}(s)\big)^{r_m}.
\end{equation}
\end{proposition}

\Cref{prop:pairwise_chernoff} provides the key analytical tool that allows tractable error control.
The Chernoff bound transforms the error probability based on the log-likelihood difference and observation sets into an exponential bound that depends only on the query counts.
Two structural properties are particularly important.
First, the bound factorizes across models: each model contributes a term $(M_m^{(y,y')}(s))^{r_m}$.
Second, each of these terms decays exponentially with the number of queries $r_m$.
Intuitively, every additional observation from an informative model reduces the probability that the competing label can accumulate sufficient evidence to beat the true label.
Besides, the parameter $s$ acts as an exponential tilting parameter that optimally balances the contributions of the two competing distributions.
Optimizing over $s$ selects the tightest exponential bound for each pairwise comparison.
Combining \cref{lem:union_bound,prop:pairwise_chernoff} leads to the following upper bound on the statewise classification error.

\begin{theorem}[Statewise Error Upper Bound]
\label{thm:statewise_chernoff}
For any query plan $r\in\mathbb{Z}_{\ge 0}^K$ and label $y\in\mathcal{Y}$,
\begin{equation}
\label{eq:statewise_chernoff}
P_e(y;r)\le \overline{P}_e(y;r)
:=
\sum_{y'\in\mathcal{Y}\setminus\{y\}}
\min_{s\in[0,1]}
\left(\frac{\pi(y')}{\pi(y)}\right)^{s}
\prod_{m=1}^K \big(M_{m}^{(y,y')}(s)\big)^{r_m}.
\end{equation}
\end{theorem}

This result provides a tractable surrogate for the otherwise intractable statewise error probability.
Beyond its analytical convenience, the bound also offers useful insights into how query effort affects classification accuracy.

First, the bound exhibits a separable structure in the query vector $r$.
For any fixed pair $(y,y')$ and tilting parameter $s$, the expression consists of a product of terms $(M_m^{(y,y')}(s))^{r_m}$, each corresponding to one model. This means that every additional query from model $m$ contributes a fixed amount of evidence toward correct classification, i.e., a multiplicative reduction in the probability that the competing label $y'$ can overtake the true label $y$.
This separable structure is exactly what makes the surrogate problem \eqref{eq:surrogate_problem} efficient: the effect of querying each model can be evaluated independently and aggregated through simple additive contributions.

Second, the bound provides a natural interpretation of model informativeness.
The quantity $M_m^{(y,y')}(s)$ captures the statistical overlap between the two distributions associated with labels $y$ and $y'$ under model $m$.
When the two distributions are very similar, the overlap is close to one, meaning that observations from this model provide little evidence for distinguishing the two labels.
In contrast, when the distributions are well separated, the overlap is small, and each observation from that model quickly reduces the probability of confusion between $y$ and $y'$.
As a result, models that are more informative for distinguishing specific label pairs will contribute larger reductions in the bound, and an optimal query strategy will naturally allocate more queries to
such models.

Finally, the bound highlights the role of prior information.
The factor $(\pi(y')/\pi(y))^s$ reflects the initial advantage that one label may have over another before any data is observed.
When the true label $y$ has a higher prior probability, the classifier begins with a favorable bias, so less observational evidence is needed to maintain correct classification.
However, when the competing label $y'$ has a stronger prior,
additional observations must compensate for this disadvantage.
In this sense, the bound explicitly captures the interplay between
prior beliefs and the information obtained from data.

Taken together, these properties explain why the bound is well suited for query design.
Although it replaces the exact error probability with an upper bound, it preserves the key statistical properties that guide the learning process: prior information, the discriminative power of each model, and the cumulative effect of additional observations.
At the same time, its separable structure allows efficient
optimization over the query plan $r$ and makes it a practical tool for designing cost-effective data collection strategies.

\subsection{Tightness of Optimization Decisions}\label{subsec:opt_tightness}
The previous subsection establishes that the surrogate error probability $\overline{P}_e(y;r)$ is a valid upper bound on the true MAP error $P_e(y;r)$. This guarantees that the surrogate problem is feasibility-preserving, but it does not by itself answer the more relevant design question: does the surrogate significantly change the optimal query plan? We note that, for our purposes, the natural notion of tightness is decision-based rather than pointwise. To make this precise, we compare the minimum costs achievable under the exact and surrogate problems.

We fix a vector of error tolerances $\alpha=(\alpha_y)_{y\in\mathcal{Y}}\in(0,1)^L$ and define the true and surrogate feasible sets
\[
\begin{aligned}
\mathcal{R}(\alpha)
&:=
\Big\{
r\in\mathbb{Z}_{\ge 0}^K:
P_e(y;r)\le \alpha_y \ \text{for all }y\in\mathcal{Y}
\Big\}, \ \overline{\mathcal{R}}(\alpha) :=
\Big\{
r\in\mathbb{Z}_{\ge 0}^K:
\overline{P}_e(y;r)\le \alpha_y \ \text{for all }y\in\mathcal{Y}
\Big\}.
\end{aligned}
\]
Using these feasible sets, we can define the optimal costs associated with each formulation.

\begin{definition}[Optimal costs]
\label{def:opt_costs}
The true and surrogate optimal costs are
\[
\mathrm{OPT}(\alpha)
:=
\min_{r\in\mathcal{R}(\alpha)} C(r),
\qquad
\overline{\mathrm{OPT}}(\alpha)
:=
\min_{r\in\overline{\mathcal{R}}(\alpha)} C(r),
\]
with the convention $\min\emptyset:=+\infty$.
\end{definition}

Whenever the relevant feasible set is nonempty, the corresponding minimum is attained because $c_m>0$ for every $m$ and each cost sublevel set $\{r\in\mathbb{Z}_{\ge 0}^K: C(r)\le c\}$ is finite.
Moreover, the statewise Chernoff bound in \Cref{thm:statewise_chernoff} implies the inclusion $\overline{\mathcal{R}}(\alpha)\subseteq \mathcal{R}(\alpha)$, and hence the following result.

\begin{corollary}[Surrogate conservatism]
\label{cor:surrogate_conservative}
For every tolerance vector $\alpha\in(0,1)^L$,
\[
\mathrm{OPT}(\alpha)\le \overline{\mathrm{OPT}}(\alpha).
\]
\end{corollary}

\cref{cor:surrogate_conservative} suggests that the surrogate problem is conservative. We point out that since the surrogate bound $\overline{P}_e(y;r)$ is obtained by combining a union bound across competing labels with pairwise Chernoff bounds, the gap between the true error $P_e(y;r)$ and the surrogate $\overline{P}_e(y;r)$ at a fixed query vector $r$ is not the most informative metric. What matters in practice is whether replacing the exact constraints by the surrogate changes the optimal cost $\mathrm{OPT}(\alpha)$ required to achieve the prescribed error tolerances.

To state our main result in this section, we impose one mild regularity condition beyond the standing assumptions of Section~\ref{s:model}. We require that the log-likelihood ratios are uniformly bounded, that is, there exists $B<\infty$ such that
\[
\left|
\log\left(\frac{p_m(x\mid y')}{p_m(x\mid y)}\right)
\right|
\le B
\qquad
\text{for all }m,\ y\neq y',\ x\in\mathcal{X}_m.
\]
This condition is standard in finite-alphabet settings and, together with the assumption that every label pair is distinguishable by at least one model, rules out degenerate cases in which a small number of observations can dominate the comparison between two labels. Let $\alpha_{\min}:=\min_{y\in\mathcal{Y}} \alpha_y$ denote the smallest error tolerance. The following result shows that the additional cost from using the surrogate vanishes as $\alpha_{\min}$ approaches 0.

\begin{theorem}[Optimization-level tightness]
\label{thm:optight}
Suppose the uniform boundedness condition above is satisfied. Then for all sufficiently small $\alpha_{\min}$, if $\mathrm{OPT}(\alpha)<\infty$,
\[
1
\le
\frac{\overline{\mathrm{OPT}}(\alpha)}{\mathrm{OPT}(\alpha)}
\le
1
+
O\!\left(
\frac{
\log\log\!\left(1/\alpha_{\min}\right)
}{
\log\!\left(1/\alpha_{\min}\right)
}\right).
\]
In particular, this ratio approaches $1$ as $\alpha_{\min}\to 0$.
\end{theorem}

The tightness result in \Cref{thm:optight} shows that the Chernoff surrogate is asymptotically exact at the optimization level. Although the surrogate may exclude some query plans that are feasible for the original problem, the resulting increase in the optimal cost $\mathrm{OPT}(\alpha)$ is negligible relative to the true optimum once the error tolerances become sufficiently small. In other words, in the high-reliability regime, solving the tractable surrogate problem yields essentially the same minimum cost as solving the intractable exact problem.

The threshold for how small  $\alpha_{\min}$ should be is a problem-dependent constant, determined by the prior, the log-likelihood bound $B$, cost structure, and the Chernoff affinity factors, but not by the tolerance vector $\alpha$ itself. It delineates the high-reliability regime in which $\mathrm{OPT}(\alpha)$ is large enough (of order $\log(1/\alpha_{\min})$) for the additive padding cost to become negligible. For any concrete instance, this threshold can in principle be computed from the model parameters. In practice, the requirement of $\alpha_{\min}$ being sufficiently small is mild, since the regime of interest is precisely when error tolerances are not big.

The proof hinges on two observations. First, for each label pair $(y,y')$, the optimized Chernoff proxy matches the true pairwise comparison probability up to a sub-exponential factor of order $1/\sqrt{n}$, where $n$ is the number of queries allocated to models that actually distinguish $y$ from $y'$. Thus the surrogate captures the correct exponential decay rate of the error and loses only a polynomial prefactor. Second, this remaining gap can be removed with very little additional query effort. Specifically, each additional round of queries to every model shrinks each pairwise Chernoff term by a fixed factor $\vartheta\in(0,1)$ (see \cref{lem:app_theta}).
Since any plan that is optimal for the original problem already requires total query effort of order $\log(1/\alpha_{\min})$, only $O(\log\log(1/\alpha_{\min}))$ additional rounds are needed to absorb the polynomial prefactor and allow the plan feasible for the surrogate. This produces a lower-order additive cost, which is precisely what drives the cost ratio in \Cref{thm:optight} to one.

We note that the practical implication is significant. The surrogate is not just a convenient analytical relaxation. It preserves the first-order cost trade-off that controls how query effort should be distributed across models. This observation makes the upper bound construction in Section~\ref{subsec:upper_bound} appropriate not only for proving feasibility, but also for guiding the actual query design. The next subsection leverages this structure to describe how the surrogate problem can be solved computationally.

\subsection{Approximation Scheme for the Surrogate Problem}
This section develops an approximation algorithm for the surrogate problem \eqref{eq:surrogate_problem}. The constraint in~\eqref{eq:surrogate_problem} involves $\overline{P}_e(y;r)$, which itself contains a minimization problem over the Chernoff parameter $s\in[0,1]$ for each pairwise term.
In what next, we provide a reformulation of the surrogate problem which removes the inner minimization over parameter $s$.
The following proposition shows that this is possible without any loss of optimality or feasibility.

\begin{proposition}[Reformulation]
\label{prop:joint_reformulation}
Let $\mathcal{P}=\{(y,y')\in\mathcal{Y}^{2}:y\neq y'\}$ denote the set of ordered pairs.
The surrogate problem~\eqref{eq:surrogate_problem} is equivalent to the joint problem
\begin{equation}\label{eq:joint_reformulation}
  \begin{aligned}
     \min_{\substack{r\in\mathbb{Z}_{\ge 0}^{K},\\[2pt]
                  s\,=\,(s_{(y,y')})_{(y,y')\in\mathcal{P}}\,\in\,[0,1]^{|\mathcal{P}|}}} &C(r) \\
    \text{s.t.} \quad
  & \sum_{y'\neq y}
  \left(\frac{\pi(y')}{\pi(y)}\right)^{\!s_{(y,y')}}
  \prod_{m=1}^{K}\!\bigl(M_{m}^{(y,y')}(s_{(y,y')})\bigr)^{r_{m}}
  \;\le\;\alpha_{y},
  \qquad\forall\,y\in\mathcal{Y}.
  \end{aligned}
\end{equation}
\end{proposition}

Proposition~\ref{prop:joint_reformulation} brings two benefits.
First, feasibility becomes log-linear in $r$ for fixed $s$: with $s$ held constant, each constraint in~\eqref{eq:joint_reformulation} reduces to a sum of exponentials that are log-linear in $r$, leading to a convex and separable subproblem amenable to standard solvers.
Second, the tilting parameters $s_{(y,y')}$ are now explicit decision variables on the compact domain $[0,1]^{|\mathcal{P}|}$, which opens the door to discretisation: approximating this domain with a finite grid of resolution $\delta$ reduces the problem to a finite collection of integer subproblems, each with fixed $s$.
 
Based on this formulation, we next propose an algorithm that solves the surrogate problem efficiently.
This algorithm follows asymptotic fully polynomial-time approximation scheme (AFPTAS).

\begin{breakablealgorithm}
\caption{AFPTAS for the Chernoff Surrogate Design Problem}
\label{alg:afptas}
\begin{algorithmic}[1]
\STATE \textbf{Input:} $(K,\mathcal{Y},\mathcal{P},\{M_m^{(y,y')}\},\{c_m\},\pi,\alpha)$,
       accuracy $\varepsilon\in(0,1]$,
       query-count bound $N_{\max}$,
       scale parameter $\Lambda$
\STATE \textbf{Output:} surrogate-feasible design $\widehat r \in \mathbb{Z}_{\ge 0}^K$
\STATE Set grid mesh $h \gets \log(1+\varepsilon)/\Lambda$
\STATE Build $s$-grid $\mathcal G_\varepsilon \gets \{0,h,2h,\dots,\lfloor 1/h\rfloor h,1\}^{|\mathcal P|}$
\FOR{each grid point $s \in \mathcal G_\varepsilon$}
  \STATE Set rounding scale $\Delta \gets \log(1+\varepsilon)/N_{\max}$
         and state-space cutoff $T_{\max} \gets \lceil B N_{\max}/\Delta \rceil$
  \FOR{$m=1,\dots,K$, $p\in\mathcal P$}
    \STATE Round discrimination weight $\widetilde w_{m,p}(s) \gets \left\lfloor -\log M_m^{(y,y')}(s_p)/\Delta \right\rfloor$
  \ENDFOR
  \STATE $\mathrm{DP}[\mathbf 0] \gets 0$
  \FOR{$t \in \{0,\dots,T_{\max}\}^{|\mathcal P|}$ in non-decreasing $\|t\|_1$ order}
    \STATE $\mathrm{DP}[t] \gets \min_m \{ c_m + \mathrm{DP}[(t-\widetilde w_m)_+] \}$
  \ENDFOR
  \STATE Find cheapest feasible state $t^\dagger(s) \gets \arg\min \{\mathrm{DP}[t] : \sum_{y'\neq y} \left(\frac{\pi(y')}{\pi(y)}\right)^{s_p} e^{-\Delta t_p}\le \alpha_y,\ \forall y\in\mathcal Y \}$
  \STATE Recover design $r^\dagger(s)$ by backtracking from $t^\dagger(s)$
\ENDFOR
\STATE $\widehat r \gets \arg\min_{s\in\mathcal G_\varepsilon} C(r^\dagger(s))$
\STATE \textbf{return} $\widehat r$
\end{algorithmic}
\end{breakablealgorithm}

Algorithm~\ref{alg:afptas} proceeds in two stages. In the first stage, it discretises the continuous tilting parameter
$s \in [0,1]^{|\mathcal{P}|}$ into a finite $s$-grid $\mathcal{G}_\varepsilon$ by setting a mesh $h = \log(1+\varepsilon)/\Lambda$, where $\Lambda$ is a scale parameter controlling how fast the pairwise Chernoff bound changes as $s$ varies. The mesh is calibrated so that adjusting any tilting parameter to the nearest grid point inflates each pairwise error bound by at most a factor of $(1+\varepsilon)$. 
In the second stage, the algorithm iterates over every
grid point $s \in \mathcal{G}_\varepsilon$ and runs a fixed-$s$ dynamic program. At each grid point, the discrimination weights
$w_{m,p}(s) = -\log M_m^{(y,y')}(s_p)$ are rounded down to integer multiples of the rounding scale $\Delta = \log(1+\varepsilon)/N_{\max}$, which produce integer weights $\widetilde{w}_{m,p}(s)$. 
Rounding down is conservative: it understates how much each query discriminates, so any design the DP allows feasible remains feasible under the exact weights. The DP itself is an
unbounded-knapsack recursion over states $t \in \{0,\ldots,T_{\max}\}^{|\mathcal{P}|}$, where $\mathrm{DP}[t]$
records the minimum cost to accumulate at least $t_p$ rounded discrimination units for every pair $p$. The cheapest state whose associated conservative error bound $\sum_{y' \neq y} \left(\frac{\pi(y')}{\pi(y)}\right)^{s_p} e^{-\Delta t_p}$ satisfies all constraints is selected as $t^\dagger(s)$, and the corresponding integer design $r^\dagger(s)$ is recovered by backtracking. The algorithm finally returns the cheapest design across all grid points.

\begin{theorem}[AFPTAS guarantee]
\label{thm:afptas}
The design $\widehat{r}$ returned by Algorithm~\ref{alg:afptas} is surrogate-feasible, and for sufficiently small $\alpha_{\min}$, the
multiplicative guarantee
\begin{equation}
  \label{eq:multiplicative}
  C(\widehat{r})
  \;\leq\;
  (1+\varepsilon)\,\overline{\mathrm{OPT}}(\alpha)
\end{equation}
holds. Algorithm~\ref{alg:afptas} is polynomial in $K$, $\log(1/\alpha_{\min})$, and $1/\varepsilon$ for fixed $L$.
\end{theorem}

Theorem~\ref{thm:afptas} follows from the discretization and conservative rounding in the two stages of Algorithm~\cref{alg:afptas}. First, discretizing the Chernoff tilting parameters onto the grid $\mathcal{G}\varepsilon$ allows that replacing any optimal tilting vector by its nearest grid point increases each pairwise Chernoff bound by at most a multiplicative factor $(1+\varepsilon)$. Second, the dynamic program uses rounded-down discrimination weights, which makes the feasibility check conservative: the rounded weights underestimate the true discrimination power of each query, so any design that is feasible by the DP also satisfies the original surrogate constraints. This rounding introduces only a bounded additive loss in the accumulated discrimination levels. 
When the minimum error tolerance $\alpha_{\min}$ is sufficiently small, the optimal cost necessarily scales at least logarithmically with $1/\alpha_{\min}$. 
As a result, the additive loss becomes negligible relative to the optimal objective value, leading to the multiplicative guarantee $C(\widehat r) \le (1+\varepsilon) \overline{\mathrm{OPT}}(\alpha)$. Overall, Algorithm~\ref{alg:afptas} constitutes an asymptotic fully polynomial-time approximation scheme: although the $(1+\varepsilon)$ approximation ratio is not uniform across all parameter regimes, it holds in the practically relevant small-tolerance regime while the running time remains polynomial in $K$, $\log(1/\alpha_{\min})$, and $1/\varepsilon$ for fixed $L$.

\section{Conclusion} \label{s:conclusion}
This paper addresses the problem of allocating queries across multiple heterogeneous LLMs to classify an unknown label at minimum cost while meeting label-specific accuracy guarantees. Three main contributions emerge from the analysis. First, the query design problem is NP-hard, and the source of hardness is the combinatorial interaction between per-model costs and the requirement that every ground-truth label should be reliably covered. Second, a Chernoff-based surrogate relaxation makes the problem tractable without meaningfully inflating the optimal cost. The surrogate-to-true cost ratio vanishes at rate $\mathcal{O}(\log\log(1/\alpha_{\min}) / \log(1/\alpha_{\min}))$, confirming that the relaxation preserves the economically relevant structure of the original problem when error tolerances are small. Third, an asymptotic fully polynomial-time approximation scheme solves the surrogate to within a $(1+\varepsilon)$ factor of the surrogate optimum.
 
Taken together, these results provide an alternative to the ad-hoc heuristics that currently govern multi-LLM query allocation. Models differ in their per-query costs and in their ability to distinguish specific label pairs, making the optimal allocation fundamentally combinatorial. Our framework shows that this difficulty can be overcome through a surrogate whose optimal plans are near-optimal for the original and can be computed efficiently. The formulation applies broadly to any setting where multiple stochastic classifiers are queried on a shared unknown state, including medical diagnosis, content categorization, and document review.

\newpage 
\bibliographystyle{informs2014}
\bibliography{multi_llm_query}

\newpage

\begin{center}
\LARGE \textbf{Online Appendices to ``Multi-LLM Query Optimization''}
\end{center}

\appendix

\section{Proofs for \cref{s:model}}
\label{app:proofs_model}
This appendix collects the proofs of all results stated in \cref{s:model}.
\subsection{Proof of \cref{thm:nphard}}
\label{proof:nphard}
\begin{proof}[Proof of \cref{thm:nphard}]
We reduce from \textsc{Minimum-Weight Set Cover}, which is
NP-hard~\citep{karp2009reducibility}.
An instance consists of a universe $U = \{1,\ldots,n\}$, sets
$S_1,\ldots,S_K \subseteq U$ with $\bigcup_{j=1}^K S_j = U$, positive
weights $w_1,\ldots,w_K > 0$, and a budget $B_{\mathrm{SC}} > 0$. The
question is whether a sub-collection of total weight at most
$B_{\mathrm{SC}}$ can cover $U$.
 
\noindent\textbf{Construction.}
Given such an instance, we build a query-plan instance as follows.
Set $L = n+1$ and $\mathcal{Y} = \{0,1,\ldots,n\}$, where label~$0$ is a
``null'' confuser and labels $1,\ldots,n$ correspond to universe elements.
Define the prior
\begin{equation}\label{appeq:prior}
  \pi(0) = \frac{1}{2},
  \qquad
  \pi(i) = \frac{1}{2n}, \quad i = 1,\ldots,n.
\end{equation}
Introduce a \emph{discriminator} model $m^*$ with output alphabet
$\{1,\ldots,n\}$, cost $c_{m^*} = \delta' > 0$ (negligibly small),
and conditional distributions
\begin{equation}\label{appeq:disc}
  p_{m^*}(x \mid y) =
  \begin{cases}
    \mathbf{1}[x = i], & y = i \ge 1, \\
    \frac{1}{n},       & y = 0.
  \end{cases}
\end{equation}
For each set $S_j$, introduce a \emph{set-cover} model~$j$ with binary
output alphabet $\{0,1\}$, cost $c_j = w_j$, and, for a fixed
$\varepsilon \in (0, \tfrac{1}{4})$,
\begin{equation}\label{appeq:scmodel}
  p_j(x \mid y) =
  \begin{cases}
    \frac{1}{2},
      & y = 0, \\
    (1-\varepsilon)^{\mathbf{1}[x=1]}\,\varepsilon^{\mathbf{1}[x=0]},
      & y = i \ge 1 \text{ and } i \in S_j, \\
    \frac{1}{2},
      & y = i \ge 1 \text{ and } i \notin S_j.
  \end{cases}
\end{equation}
In words, when $i \in S_j$, model~$j$ outputs~$1$ with probability
$1-\varepsilon$ and~$0$ with probability $\varepsilon$. Otherwise, it
outputs each value with probability $1/2$, identically to its
behavior under label~$0$.
Set tolerances $\alpha_0 = 1 - \delta''$ (nearly trivial, for small
$\delta'' > 0$) and $\alpha_i = 2\varepsilon$ for $i = 1,\ldots,n$, and
budget $B = B_{\mathrm{SC}} + \delta'$.
 
\noindent\textbf{Error analysis.}
We show that the error constraints are equivalent to the covering
requirement. Any cost-optimal plan includes model~$m^*$ exactly once,
since its cost is negligible and it is necessary for separating
non-null labels from one another. Fix the true label $Y = i$ for some
$i \ge 1$, and suppose $r_{m^*} = 1$. Since
$p_{m^*}(i \mid j) = \mathbf{1}[i = j] = 0$ for $j \ge 1$, $j \neq i$,
the observation $X_{m^*,1} = i$ drives all non-null posteriors other
than $\Pr(Y = i \mid X_{m^*,1} = i)$ to zero, i.e., the posterior $\Pr(Y = k \mid X_{m^*,1}=i) = 0$ for every non-null label $k \neq i$.
The log-posterior-odds between labels $0$ and $i$ are then
\[
  \log\frac{\pi(0)\,p_{m^*}(i \mid 0)}{\pi(i)\,p_{m^*}(i \mid i)}
  = \log\frac{\frac{1}{2} \cdot \frac{1}{n}}{\frac{1}{2n} \cdot 1} = 0,
\]
so after model~$m^*$ the MAP decision reduces to a binary comparison
between labels $i$ and $0$ starting from a perfect tie.
 
By inspection of~\eqref{appeq:scmodel}, if $i \notin S_j$ then
$p_j(\cdot \mid i) = p_j(\cdot \mid 0) = (1/2, 1/2)$,
so model~$j$ yields log-likelihood ratio zero for every observation and
is entirely uninformative for breaking this tie. If instead $i \in S_j$,
the per-observation log-likelihood ratio is
\[
  \Lambda_j(X) := \log\frac{p_j(X \mid i)}{p_j(X \mid 0)} =
  \begin{cases}
    \log\bigl(2(1-\varepsilon)\bigr) > 0, & X = 1, \\[4pt]
    \log(2\varepsilon) < 0,               & X = 0.
  \end{cases}
\]
The cumulative log-posterior-odds after all queries are therefore
\[
  \Lambda_{\mathrm{total}}
  = \underbrace{0}_{\text{after }m^*}
  + \sum_{\substack{j:\,i \in S_j \\ r_j \ge 1}} \sum_{t=1}^{r_j}
    \Lambda_j(X_{j,t})
  + \underbrace{\sum_{\substack{j:\,i \notin S_j \\ r_j \ge 1}}
    \sum_{t=1}^{r_j} 0}_{=\,0}.
\]
Suppose there exists $j$ with $i \in S_j$ and $r_j \ge 1$. A
misclassification requires $\Lambda_{\mathrm{total}} \le 0$. With a
single covering query,
\[
  P_e(i;\,r)
  \le \Pr\!\bigl(\Lambda_j(X_{j,1}) \le 0 \mid Y = i\bigr)
  = \Pr(X_{j,1} = 0 \mid Y = i)
  = \varepsilon < 2\varepsilon = \alpha_i.
\]
Additional covering queries can only decrease $P_e(i;\,r)$, since each
contributes a strictly positive expected log-likelihood ratio favoring
the correct label. Now suppose instead that $r_j = 0$ for every $j$ with
$i \in S_j$. Every queried set-cover model satisfies $i \notin S_j$, so
$\Lambda_{\mathrm{total}} = 0$ for every realization of $X(r)$, and
labels $i$ and $0$ have identical posteriors regardless of the data.
Any deterministic tie-breaking rule that does not uniformly favor
label~$i$ yields $P_e(i;\,r) = 1 > 2\varepsilon$; the best randomized
rule (uniform) yields $P_e(i;\,r) = 1/2 > 2\varepsilon$ since
$\varepsilon < 1/4$. Combining both directions: the statewise
constraint $P_e(i;\,r) \le 2\varepsilon$ holds for every $i =
1,\ldots,n$ if and only if $\mathcal{C}(r) := \{j : r_j \ge 1\}$ covers
$U$.
 
\noindent\textbf{Cost equivalence.}\;
Given a set cover $\mathcal{C}^*$ with
$\sum_{j \in \mathcal{C}^*} w_j \le B_{\mathrm{SC}}$, define
$r_j = \mathbf{1}[j \in \mathcal{C}^*]$ and $r_{m^*} = 1$.
This plan is feasible by the above equivalence and has cost
$C(r) = \delta' + \sum_{j \in \mathcal{C}^*} w_j \le B$.
By contrast, given a feasible plan $r$ with $C(r) \le B$, we may assume
without loss of generality that $r_j \in \{0,1\}$ for each set-cover
model (additional queries only increase cost without avoiding feasibility),
so $\mathcal{C}(r)$ is a valid cover with
$\sum_{j \in \mathcal{C}(r)} w_j = C(r) - \delta' \le B_{\mathrm{SC}}$.
 
 
The construction is polynomial-time and the two decision problems have
identical yes/no answers, so NP-hardness of \textsc{Minimum-Weight Set
Cover} implies NP-hardness of the decision problem~\eqref{eq:opt_problem}.
\end{proof}

\section{Proofs for \cref{s:algorithm}}
\label{app:proofs_alg}

This appendix collects the proofs of all results stated in \cref{s:algorithm}.

\subsection{Proof of \cref{lem:union_bound}}
\label{proof:union_bound}

\begin{proof}[Proof of \cref{lem:union_bound}]

\textbf{MAP rule as posterior likelihood maximization.}
By assumption, all likelihoods are strictly positive, and
$\pi(a)>0$ for all $a$.
Therefore the posterior quantity
\[
\pi(a)\prod_{m=1}^K \prod_{t=1}^{r_m} p_m(x_{m,t}\mid a)
\]
is strictly positive for every $a$ and every realization $x(r)$.
Since the logarithm is strictly increasing, maximizing the posterior
probability is equivalent to maximizing the posterior log-likelihood,
and thus
\[
\widehat{Y}_r(x(r))
\in
\arg\max_{a\in\mathcal{Y}}
\left[
\log \pi(a)
+
\sum_{m=1}^{K}\sum_{t=1}^{r_m}
\log p_m(x_{m,t}\mid a)
\right].
\]

\noindent
\textbf{Error event.}
Fix the ground-truth label $y$ and define the error event
\[
\mathcal{E}_y=\{\widehat{Y}_r(X(r))\neq y\}.
\]

If $\mathcal{E}_y$ occurs, then there must exist at least one competing
label $y'\neq y$ such that
\[
\log \pi(y')
+
\sum_{m,t}\log p_m(X_{m,t}\mid y')
\ge
\log \pi(y)
+
\sum_{m,t}\log p_m(X_{m,t}\mid y).
\]

Rearranging leads to
\[
\Delta_{y,y'}(r)\ge 0.
\]

Therefore,
\[
\mathcal{E}_y
\subseteq
\bigcup_{y'\in\mathcal{Y}\setminus\{y\}}
\{\Delta_{y,y'}(r)\ge 0\}.
\]

\noindent
\textbf{Apply the union bound.}
Taking conditional probabilities given $Y=y$ and applying the union bound allows \eqref{eq:union_bound}.
\end{proof}

\subsection{Proof of \cref{prop:pairwise_chernoff}}
\label{proof:pairwise_chernoff}

The proof applies the exponential Markov inequality and evaluates the
resulting moment via an auxiliary factorization lemma.

\begin{lemma}[Exponential moment factorization]
\label{lem:exp_moment_factorization}
For all distinct $y,y'\in\mathcal{Y}$,
$r\in\mathbb{Z}_{\ge0}^K$, and $s\in[0,1]$,
\begin{equation}
\label{eq:exp_moment_factorization}
\E\!\left[
\exp\!\left(
s\sum_{m=1}^{K}\sum_{t=1}^{r_m}
\log
\frac{p_m(X_{m,t}\mid y')}{p_m(X_{m,t}\mid y)}
\right)
\middle| Y=y
\right]
=
\prod_{m=1}^{K}
\left(M_m^{(y,y')}(s)\right)^{r_m}.
\end{equation}

\end{lemma}

\begin{proof}[Proof of \cref{lem:exp_moment_factorization}]

Rewrite the exponent as
\begin{align*}
&\exp\!\left(
s\sum_{m=1}^K\sum_{t=1}^{r_m}
\log \left(\frac{p_m(X_{m,t}\mid y')}{p_m(X_{m,t}\mid y)}\right)
\right)
\\
&\qquad
=
\prod_{m=1}^K\prod_{t=1}^{r_m}
\left(\frac{p_m(X_{m,t}\mid y')}{p_m(X_{m,t}\mid y)}\right)^{s}
=
\prod_{m=1}^K\prod_{t=1}^{r_m}
p_m(X_{m,t}\mid y)^{-s}\,p_m(X_{m,t}\mid y')^{s}.
\end{align*}

Conditional on $Y=y$ the random
variables $\{X_{m,t}\}_{m,t}$ are independent across $(m,t)$, and for
each fixed $m$ the repeats are i.i.d.\ with $p_m(\cdot\mid y)$.
Therefore, by the definition of independence,
\begin{align*}
&\E\!\left[
\prod_{m=1}^K\prod_{t=1}^{r_m}
p_m(X_{m,t}\mid y)^{-s}\,p_m(X_{m,t}\mid y')^{s}
\,\bigg|\,
Y=y
\right]
\\
&\qquad
=
\prod_{m=1}^K\prod_{t=1}^{r_m}
\E\!\left[
p_m(X_{m,t}\mid y)^{-s}\,p_m(X_{m,t}\mid y')^{s}
\,\bigg|\,
Y=y
\right].
\end{align*}
For any fixed $(m,t)$, conditional on $Y=y$ we have
$\Pr(X_{m,t}=x\mid Y=y)=p_m(x\mid y)$, hence
\begin{align*}
\E\left[
p_m(X_{m,t}\mid y)^{-s}\,p_m(X_{m,t}\mid y')^{s}
\,\bigg|\,
Y=y
\right]
&=
\sum_{x\in \mathcal{X}_m}
p_m(x\mid y)\,p_m(x\mid y)^{-s}\,p_m(x\mid y')^{s}
\\
&=
\sum_{x\in\mathcal{X}_m}
p_m(x\mid y)^{\,1-s}\,p_m(x\mid y')^{\,s}
=
M_{m}^{(y,y')}(s),
\end{align*}
where the last equality is \eqref{eq:M_def}.
Since the same quantity appears for each $t\in\{1,\ldots,r_m\}$, we
obtain \eqref{eq:exp_moment_factorization}.
\end{proof}

We next establish basic analytic properties of the Chernoff affinity factor $M_m^{(y,y')}(s)$ used throughout the analysis.

\begin{lemma}[Bounds on the Chernoff affinity factor]
\label{lem:M_bounds}
Under Assumption~3.2, for all $s\in[0,1]$,
\[
0 < M_m^{(y,y')}(s) \le 1.
\]
Moreover,
\[
M_m^{(y,y')}(0)=M_m^{(y,y')}(1)=1.
\]
For $s\in(0,1)$, equality $M_m^{(y,y')}(s)=1$ holds if and only if
$p_m(\cdot\mid y)=p_m(\cdot\mid y')$.
\end{lemma}

\begin{proof}[Proof of \cref{lem:M_bounds}]
\textit{Positivity.}
Under Assumption~3.2, each term
\[
p_m(x\mid y)^{1-s}p_m(x\mid y')^{s}
\]
is strictly positive.
Summing over the finite set $\mathcal{X}_m$ therefore gives
$M_m^{(y,y')}(s)>0$.

\noindent
\textbf{Endpoint values.}
At $s=0$,
\[
M_m^{(y,y')}(0)=\sum_{x\in\mathcal{X}_m} p_m(x\mid y)=1.
\]
At $s=1$,
\[
M_m^{(y,y')}(1)=\sum_{x\in\mathcal{X}_m} p_m(x\mid y')=1.
\]

\noindent
\textbf{Upper bound for $s\in(0,1)$ via H\"older's inequality.}
Let
\[
f(x)=p_m(x\mid y)^{1-s},
\qquad
g(x)=p_m(x\mid y')^{s}.
\]
Choose conjugate exponents
\[
p=\frac{1}{1-s}, \qquad q=\frac{1}{s}.
\]
Applying H\"older's inequality on the finite sum over $\mathcal{X}_m$
gives
\[
\sum_{x} f(x)g(x)
\le
\left(\sum_{x} f(x)^p\right)^{1/p}
\left(\sum_{x} g(x)^q\right)^{1/q}.
\]

Since
\[
f(x)^p=p_m(x\mid y),
\qquad
g(x)^q=p_m(x\mid y'),
\]
both sums equal 1 by normalization.
Hence $M_m^{(y,y')}(s)\le 1$.

\noindent
\textbf{Equality condition.}
H\"older's inequality is tight if and only if
$f(x)^p$ is proportional to $g(x)^q$ for all $x$.
This requires
\[
p_m(x\mid y)=\gamma\,p_m(x\mid y')
\]
for some $\gamma>0$.
Summing over $x$ gives $\gamma=1$, which implies
$p_m(\cdot\mid y)=p_m(\cdot\mid y')$.
\end{proof}


\begin{proof}[Proof of \cref{prop:pairwise_chernoff}]
Fix $s\in[0,1]$. By definition, we have
\[
\Delta_{y,y'}(r)
=
\log\!\left(\frac{\pi(y')}{\pi(y)}\right)
+
\sum_{m=1}^K\sum_{t=1}^{r_m}
\log\!\left(\frac{p_m(X_{m,t}\mid y')}{p_m(X_{m,t}\mid y)}\right).
\]

\noindent
\textbf{Apply Markov's inequality.}
For $s>0$, the map $u\mapsto e^{su}$ is strictly increasing, so
$\Delta_{y,y'}(r)\ge 0$ implies $e^{s\,\Delta_{y,y'}(r)}\ge e^0=1$.
For $s=0$, $e^{s\,\Delta_{y,y'}(r)}=1\ge 1$ holds trivially.
Therefore, for every $s\in[0,1]$,
\[
\{\Delta_{y,y'}(r)\ge 0\}
\subseteq
\{e^{s\,\Delta_{y,y'}(r)}\ge 1\}.
\]
By monotonicity of probability and Markov's inequality applied to the
nonnegative random variable $e^{s\,\Delta_{y,y'}(r)}$ under
$\Pr(\,\cdot\mid Y=y)$,
\begin{equation}
\label{eq:markov_step}
\Pr\!\big(\Delta_{y,y'}(r)\ge 0 \,\big|\, Y=y\big)
\le
\Pr\!\big(e^{s\,\Delta_{y,y'}(r)}\ge 1 \,\big|\, Y=y\big)
\le
\E\!\big[e^{s\,\Delta_{y,y'}(r)} \,\big|\, Y=y\big].
\end{equation}

\noindent
\textbf{Compute the conditional expectation.}
Using the decomposition of $\Delta_{y,y'}(r)$ and the fact that
$e^{s(a+b)}=e^{sa}e^{sb}$, we obtain
\[
e^{s\Delta_{y,y'}(r)}
=
\left(\frac{\pi(y')}{\pi(y)}\right)^{s}
\exp\!\left(
s\sum_{m=1}^K\sum_{t=1}^{r_m}
\log\!\left(\frac{p_m(X_{m,t}\mid y')}{p_m(X_{m,t}\mid y)}\right)
\right).
\]
Taking conditional expectations given $Y=y$ yields
\[
\E\!\left[e^{s\Delta_{y,y'}(r)} \,\big|\, Y=y\right]
=
\left(\frac{\pi(y')}{\pi(y)}\right)^{s}
\E\!\left[
\exp\!\left(
s\sum_{m=1}^K\sum_{t=1}^{r_m}
\log\!\left(\frac{p_m(X_{m,t}\mid y')}{p_m(X_{m,t}\mid y)}\right)
\right)
\,\bigg|\,
Y=y
\right].
\]
By Lemma~\ref{lem:exp_moment_factorization}, the expectation factorizes
as $\prod_m (M_m^{(y,y')}(s))^{r_m}$.
Substituting into \eqref{eq:markov_step} and minimizing over $s\in[0,1]$ proves \eqref{eq:pairwise_chernoff_opt}.
\end{proof}
\subsection{Proof of \cref{thm:statewise_chernoff}}
\label{proof:statewise_chernoff}

\begin{proof}[Proof of \cref{thm:statewise_chernoff}]

Fix $y\in\mathcal{Y}$ and $r\in\mathbb{Z}_{\ge0}^{K}$.
By \cref{lem:union_bound},
\[
P_e(y;r)
\le
\sum_{y'\in\mathcal{Y}\setminus\{y\}}
\Pr\!\big(\Delta_{y,y'}(r)\ge0 \mid Y=y\big).
\]

Applying \cref{prop:pairwise_chernoff} to each term gives
\[
P_e(y;r)
\le
\sum_{y'\in\mathcal{Y}\setminus\{y\}}
\min_{s\in[0,1]}
\left(\frac{\pi(y')}{\pi(y)}\right)^s
\prod_{m=1}^{K}
\left(M_m^{(y,y')}(s)\right)^{r_m}.
\]

This expression equals $\overline{P}_e(y;r)$ defined in
\eqref{eq:statewise_chernoff}, completing the proof.
\end{proof}

\subsection{Proof of \cref{thm:optight}}
\label{app:proof_opt_tightness}

Throughout this subsection, we use the shorthand
$\pi_{\min}:=\min_{y}\pi(y)$,
$\pi_{\max}:=\max_{y}\pi(y)$,
$c_{\min}:=\min_{m} c_m$,
and $c_{\Sigma}:=\sum_{m=1}^K c_m$.

To prove Theorem \ref{thm:optight}, we proceed in two stages.  First, we establish a \emph{pairwise tightness} result showing that the
Chernoff bound for each pair $(y,y')$ matches the true pairwise
comparison probability up to a factor of order $1/\sqrt{n}$, where $n$
is the effective sample size for that pair.
Second, we translate this pairwise tightness into a bound on the ratio
$\overline{\mathrm{OPT}}(\alpha)/\mathrm{OPT}(\alpha)$ by showing that
only $O(\log\log(1/\alpha_{\min}))$ additional rounds of queries suffice
to absorb the polynomial prefactor.

\subsubsection{Pairwise Chernoff tightness}

We work with a \emph{strict} pairwise comparison event:
for $y\neq y'$ and $r\in\mathbb{Z}_{\ge 0}^K$, define
$E^{>}_{y,y'}(r):=\{\Delta_{y,y'}(r)>0\}$.
Since $\Delta_{y,y'}(r)>0$ implies that
$y'$ achieves a strictly larger posterior log-likelihood than $y$,
we have
$E^{>}_{y,y'}(r)\subseteq \{\widehat{Y}_r(X(r))\neq y\}$ and hence
$\Pr(E^{>}_{y,y'}(r)\mid Y=y)\le P_e(y;r)$.

For each pair $(y,y')$ and $s\in[0,1]$, define the \emph{pairwise Chernoff proxy}
\begin{equation}
\label{eq:app_pairwise_proxy_def}
\overline{P}_{y,y'}(r;s)
:=
\left(\frac{\pi(y')}{\pi(y)}\right)^{s}
\prod_{m=1}^K \big(M_{m}^{(y,y')}(s)\big)^{r_m},
\end{equation}
and the optimized proxy
$\overline{P}_{y,y'}(r):=\min_{s\in[0,1]} \overline{P}_{y,y'}(r;s)$.
From \cref{prop:pairwise_chernoff}, we have
$\Pr(\Delta_{y,y'}(r)\ge 0\mid Y=y)\le \overline{P}_{y,y'}(r)$
for all $y\neq y'$ and all $r$.

We also define the set of \emph{informative models}
$\mathcal{M}(y,y'):=\{m: p_m(\cdot\mid y)\neq p_m(\cdot\mid y')\}$
and the \emph{effective sample size} $n_{y,y'}(r):=\sum_{m\in\mathcal{M}(y,y')} r_m$.
By the pairwise identifiability condition, $\mathcal{M}(y,y')\neq\emptyset$ for every $y\neq y'$.

The key analytical tool is an exponential change of measure that recasts tail probabilities of $\Delta_{y,y'}(r)$ under the original model as expectations under a tilted distribution.

\begin{remark}[Change of measure]
\label{rem:app_change_of_measure}
For $s\in[0,1]$, define a tilted marginal on $\mathcal{X}_m$ by
\[
q_m^{(y,y')}(x;s)
:=
\frac{
p_m(x\mid y)^{\,1-s}\,p_m(x\mid y')^{\,s}
}{
M_m^{(y,y')}(s)
},
\qquad x\in\mathcal{X}_m,
\]
and the induced product measure
$\mathbb{Q}_{y,y',r,s}(X(r)=x(r)):=\prod_{m,t} q_m^{(y,y')}(x_{m,t};s)$.
Since all likelihoods are strictly positive, this is a well-defined probability measure.
\end{remark}

The following identity expresses probabilities under the original model in terms of the tilted measure $\mathbb{Q}_{y,y',r,s}$.

\begin{lemma}[Change-of-measure identity]
\label{lem:app_rn_identity}
For any event $\mathcal{E}$ measurable with respect to $X(r)$,
\[
\Pr\!\big(X(r)\in\mathcal{E}\mid Y=y\big)
=
\overline{P}_{y,y'}(r;s)\,
\E_{\mathbb{Q}_{y,y',r,s}}
\Big[
e^{-s\Delta_{y,y'}(r)}\,
\mathbf{1}\{X(r)\in\mathcal{E}\}
\Big].
\]
\end{lemma}

\begin{proof}[Proof of \cref{lem:app_rn_identity}]
For a fixed realization $x(r)$, the Radon--Nikodym derivative equals
\[
\frac{
\Pr(X(r)=x(r)\mid Y=y)
}{
\mathbb{Q}_{y,y',r,s}(X(r)=x(r))
}
=
\prod_{m,t}
\frac{p_m(x_{m,t}\mid y)}{q_m^{(y,y')}(x_{m,t};s)}
=
\prod_{m,t}
M_m^{(y,y')}(s)\,
\left(\frac{p_m(x_{m,t}\mid y)}{p_m(x_{m,t}\mid y')}\right)^{s}.
\]
Using the definition of $\Delta_{y,y'}(r)$ in \eqref{eq:delta_expanded},
the product simplifies to
$\overline{P}_{y,y'}(r;s)\,e^{-s\Delta_{y,y'}(r)}$.
Multiplying by $\mathbb{Q}_{y,y',r,s}(X(r)=x(r))$ and summing over $x(r)\in\mathcal{E}$ yields the result.
\end{proof}

To apply \cref{lem:app_rn_identity} effectively, we need the tilting parameter $s$ to be chosen so that $\Delta_{y,y'}(r)$ is centered under the tilted measure.

\begin{remark}[Interior minimizers and centering]
\label{rem:app_centering}
Define $\Psi_{y,y',r}(s):=\log \overline{P}_{y,y'}(r;s)$ for $s\in[0,1]$.
This function is convex (as a sum of log-sum-exp functions and a linear term in $s$),
and its derivative satisfies
$\Psi'(s)=\E_{\mathbb{Q}_{y,y',r,s}}[\Delta_{y,y'}(r)]$
for $s\in(0,1)$.
\end{remark}

The next lemma confirms that, for large enough effective sample sizes, such a centering minimizer exists and is bounded away from the endpoints of $[0,1]$.

\begin{lemma}[Interior minimizers]
\label{lem:app_centering}
Under the model assumptions in Section~\ref{s:model} and the uniform boundedness condition,
there exist constants $\delta\in(0,1/2)$ and $n_{\mathrm{cen}}\in\mathbb{N}$
such that for any $y\neq y'$ and any $r$ with $n_{y,y'}(r)\ge n_{\mathrm{cen}}$,
every minimizer $s^*_{y,y'}(r)$ of $\Psi_{y,y',r}$ lies in $[\delta,1-\delta]$ and satisfies
\[
\E_{\mathbb{Q}_{y,y',r,s^*_{y,y'}(r)}}\!\big[\Delta_{y,y'}(r)\big]=0.
\]
\end{lemma}

\begin{proof}[Proof of \cref{lem:app_centering}]
Define
$\kappa_{\min}:=\min_{y\neq y',\,m\in\mathcal{M}(y,y')}
\min\{D_{\mathrm{KL}}(p_m(\cdot\mid y)\|p_m(\cdot\mid y')),\,
D_{\mathrm{KL}}(p_m(\cdot\mid y')\|p_m(\cdot\mid y))\}>0$.
At $s=0$, we have
$\Psi'(0)=\log(\pi(y')/\pi(y))-\sum_m r_m D_{\mathrm{KL}}(p_m(\cdot\mid y)\|p_m(\cdot\mid y'))
\le \log(1/\pi_{\min})-\kappa_{\min}\,n_{y,y'}(r)$,
and similarly $\Psi'(1)\ge -\log(1/\pi_{\min})+\kappa_{\min}\,n_{y,y'}(r)$.
For $n_{y,y'}(r)$ large enough, $\Psi'(0)<0<\Psi'(1)$, so every minimizer lies in $(0,1)$.

Set $\delta:=\min\{\kappa_{\min}/(4B^2),\,1/4\}$.
By the mean value theorem and $|\Psi''(s)|\le B^2 n_{y,y'}(r)$,
we get $\Psi'(\delta)<0<\Psi'(1-\delta)$ for $n_{y,y'}(r)$ large enough,
so every minimizer lies in $[\delta,1-\delta]$.
Since $\Psi$ is differentiable and convex on the interior, any interior minimizer satisfies $\Psi'(s^*)=0$,
which is exactly the centering condition.
\end{proof}

Under the centering tilt, $\Delta_{y,y'}(r)$ has mean zero; we next establish that its variance grows linearly in the effective sample size.

\begin{lemma}[Uniform variance lower bound under the tilted measure]
\label{lem:app_variance_lb}
Let $\delta\in(0,1/2)$ be as in \cref{lem:app_centering}.
There exists a constant $v_{\min}>0$ such that for all $y\neq y'$, all $r$, and all $s\in[\delta,1-\delta]$,
\[
\mathrm{Var}_{\mathbb{Q}_{y,y',r,s}}\!\big(\Delta_{y,y'}(r)\big)
\ge
v_{\min}\,n_{y,y'}(r).
\]
\end{lemma}

\begin{proof}[Proof of \cref{lem:app_variance_lb}]
For each informative model $m\in\mathcal{M}(y,y')$ and each $s\in[\delta,1-\delta]$, the tilted marginal $q_m^{(y,y')}(\cdot;s)$ has full support on $\mathcal{X}_m$, and the log-likelihood ratio $\log(p_m(X\mid y')/p_m(X\mid y))$ is non-constant.
Its variance under $q_m^{(y,y')}(\cdot;s)$ is therefore strictly positive.
By continuity of the variance in $s$ and finiteness of the set of informative triples $(m,y,y')$, the minimum over all such triples and all $s\in[\delta,1-\delta]$ is attained and positive.
Since $\Delta_{y,y'}(r)$ is a sum of $n_{y,y'}(r)$ independent terms (one per query to an informative model), the total variance is at least $v_{\min}\,n_{y,y'}(r)$.
\end{proof}

With the centering property from \cref{lem:app_centering} and the variance lower bound from \cref{lem:app_variance_lb} in hand, we apply the Berry--Esseen theorem to show that a non-negligible fraction of the tilted probability mass falls in a window just above zero.

\begin{lemma}[Window probability under the minimizing tilt]
\label{lem:app_window_prob}
There exist constants $H>0$, $n_{\mathrm{win}}\in\mathbb{N}$, and
$\beta>0$ such that for every $y\neq y'$ and every $r$ with
$n_{y,y'}(r)\ge n_{\mathrm{win}}$,
\[
\mathbb{Q}_{y,y',r,s^\star_{y,y'}(r)}
\big(0< \Delta_{y,y'}(r)\le H\big)
\ge
\frac{\beta}{\sqrt{n_{y,y'}(r)}}.
\]
\end{lemma}

\begin{proof}[Proof of \cref{lem:app_window_prob}]
Write $n:=n_{y,y'}(r)$, $s^\star:=s^\star_{y,y'}(r)$, $\mathbb{Q}:=\mathbb{Q}_{y,y',r,s^\star}$, and $\Delta:=\Delta_{y,y'}(r)$.
For each informative model $m\in\mathcal{M}(y,y')$, define
$Z_{m,t}:=\log(p_m(X_{m,t}\mid y')/p_m(X_{m,t}\mid y))$
and the centered variables $U_{m,t}:=Z_{m,t}-\E_{\mathbb{Q}}[Z_{m,t}]$.
Let $S:=\sum_{m\in\mathcal{M}(y,y')}\sum_{t=1}^{r_m} U_{m,t}$.
Since $\E_{\mathbb{Q}}[\Delta]=0$ (\cref{lem:app_centering}), we have $\Delta=S$.
Write $\sigma^2:=\mathrm{Var}_{\mathbb{Q}}(S)$.
By the uniform boundedness condition, $|U_{m,t}|\le 2B$ almost surely
and $\E_{\mathbb{Q}}[|U_{m,t}|^3]\le (2B)^3$.
Moreover, $v_{\min}\,n\le \sigma^2\le B^2 n$.

By the Berry--Esseen theorem for independent random variables,
$|\mathbb{Q}(S/\sigma\le z)-\Phi(z)|\le C_0/\sqrt{n}$
where $C_0:=C_{\mathrm{BE}}(2B)^3/v_{\min}^{3/2}$ and $\Phi$ is the standard normal cdf.

Choose $H:=\max\{1,\,4C_0 B/\varphi(1)\}$ and
$n_{\mathrm{win}}:=\max\{n_{\mathrm{cen}},\,\lceil H^2/v_{\min}\rceil\}$,
where $\varphi(z):=(2\pi)^{-1/2}e^{-z^2/2}$.
For $n\ge n_{\mathrm{win}}$, we have $\sigma\ge H$ and thus $H/\sigma\le 1$, so
\[
\Phi(H/\sigma)-\Phi(0)
=\int_0^{H/\sigma}\varphi(u)\,du
\ge \frac{H}{\sigma}\,\varphi(1)
\ge \frac{H\,\varphi(1)}{B\sqrt{n}}.
\]
Therefore
$\mathbb{Q}(0<\Delta\le H)\ge (H\varphi(1)/B-2C_0)/\sqrt{n}\ge 2C_0/\sqrt{n}$
by the choice of $H$.
\end{proof}

Combining the change-of-measure identity with the window probability estimate yields the main pairwise result.

\begin{lemma}[Pairwise Chernoff tightness]
\label{lem:app_pairwise_tightness}
Under the model assumptions in Section~\ref{s:model} and the uniform boundedness condition,
there exist constants $n_0\in\mathbb{N}$ and $\gamma>0$ such that for every
$y\neq y'$ and every $r$ with $n_{y,y'}(r)\ge n_0$,
\begin{equation}
\label{eq:app_pairwise_tightness}
\Pr\!\big(E^{>}_{y,y'}(r)\mid Y=y\big)
\ge
\frac{\gamma}{\sqrt{n_{y,y'}(r)}}\,
\overline{P}_{y,y'}(r).
\end{equation}
Together with \cref{prop:pairwise_chernoff}, this shows that the
pairwise Chernoff bound is tight up to a factor of order
$1/\sqrt{n_{y,y'}(r)}$.
\end{lemma}

\begin{proof}[Proof of \cref{lem:app_pairwise_tightness}]
Set $n_0:=n_{\mathrm{win}}$ and fix $y\neq y'$ and $r$ with $n:=n_{y,y'}(r)\ge n_0$.
Let $s^\star:=s^\star_{y,y'}(r)$.
Applying \cref{lem:app_rn_identity} with $\mathcal{E}=E^{>}_{y,y'}(r)$ gives
\[
\Pr\!\big(E^{>}_{y,y'}(r)\mid Y=y\big)
=
\overline{P}_{y,y'}(r;s^\star)\,
\E_{\mathbb{Q}_{y,y',r,s^\star}}
\Big[
e^{-s^\star\Delta_{y,y'}(r)}\,
\mathbf{1}\{E^{>}_{y,y'}(r)\}
\Big].
\]
On $\{0<\Delta_{y,y'}(r)\le H\}\subseteq E^{>}_{y,y'}(r)$,
we have $e^{-s^\star\Delta}\ge e^{-H}$ since $s^\star\in[0,1]$.
Substituting $\overline{P}_{y,y'}(r;s^\star)=\overline{P}_{y,y'}(r)$ and applying \cref{lem:app_window_prob}, we obtain
\[
\Pr\!\big(E^{>}_{y,y'}(r)\mid Y=y\big)
\ge
\overline{P}_{y,y'}(r)\,e^{-H}\,\frac{\beta}{\sqrt{n}}.
\]
Thus \eqref{eq:app_pairwise_tightness} holds with $\gamma:=\beta e^{-H}$.
\end{proof}

\subsubsection{From pairwise tightness to the cost ratio bound}

For $n\in\mathbb{Z}_{\ge 0}$, define the \emph{$n$-round augmentation} $r^{+n}:=r+n\mathbf{1}$ where $\mathbf{1}=(1,\ldots,1)\in\mathbb{Z}_{\ge 0}^K$.
The following lemma shows that each additional round contracts the pairwise Chernoff proxy by a fixed factor.

\begin{lemma}[Geometric contraction]
\label{lem:app_theta}
Let $\delta$ be as in \cref{lem:app_centering} and define
\[
\vartheta
:=
\max_{\substack{y\neq y'\\ s\in[\delta,1-\delta]}}
\prod_{m=1}^K M_m^{(y,y')}(s).
\]
Then $\vartheta\in(0,1)$.
Moreover, for any $y\neq y'$, any $r$, any $n\ge 0$,
and any $s\in[\delta,1-\delta]$,
\begin{equation}
\label{eq:app_theta_decay}
\overline{P}_{y,y'}(r^{+n};s)
\le
\vartheta^{n}\,\overline{P}_{y,y'}(r;s).
\end{equation}
\end{lemma}

\begin{proof}[Proof of \cref{lem:app_theta}]
For $s\in(\delta,1-\delta)\subset(0,1)$ and any pair $y\neq y'$,
H\"older's inequality gives $M_m^{(y,y')}(s)\le 1$ with equality only if
$p_m(\cdot\mid y)=p_m(\cdot\mid y')$.
Since at least one model distinguishes every pair,
$\prod_m M_m^{(y,y')}(s)<1$.
By continuity and compactness, the maximum $\vartheta$ over all pairs
and all $s\in[\delta,1-\delta]$ is attained and lies in $(0,1)$.
The inequality \eqref{eq:app_theta_decay} follows from
$\overline{P}_{y,y'}(r^{+n};s)=\overline{P}_{y,y'}(r;s)\prod_m M_m^{(y,y')}(s)^n
\le \overline{P}_{y,y'}(r;s)\,\vartheta^n$.
\end{proof}

The pairwise tightness bound in \cref{lem:app_pairwise_tightness} requires $n_{y,y'}(r)\ge n_0$.
The next lemma shows that, for sufficiently small tolerances, any robustly feasible design automatically satisfies this condition.

\begin{lemma}[Small tolerances force large effective sample sizes]
\label{lem:app_min_effective_sample}
Define
$\eta_{\min}:=\min_{m,y,x} p_m(x\mid y)>0$
and
$\varepsilon_0:=\eta_{\min}^{\,n_0-1}$,
where $n_0$ is as in \cref{lem:app_pairwise_tightness}.
For every pair $y\neq y'$ and every $r$ with $n_{y,y'}(r)\le n_0-1$,
\[
\max\big\{P_e(y;r),\,P_e(y';r)\big\}\ge \varepsilon_0.
\]
Consequently, if $\alpha_{\min}<\varepsilon_0$ and $r\in\mathcal{R}(\alpha)$,
then $n_{y,y'}(r)\ge n_0$ for all $y\neq y'$.
\end{lemma}

\begin{proof}[Proof of \cref{lem:app_min_effective_sample}]
If $n_{y,y'}(r)=0$, no informative model is queried, so the score difference
$S_{y'}-S_y=\log(\pi(y')/\pi(y))$ is deterministic.
This forces one of the two labels to be always misclassified, giving
$\max\{P_e(y;r),P_e(y';r)\}=1\ge \varepsilon_0$.

If $1\le n_{y,y'}(r)\le n_0-1$, there exists an informative model $m^\star\in\mathcal{M}(y,y')$
with $r_{m^\star}\ge 1$.
Because $p_{m^\star}(\cdot\mid y)\neq p_{m^\star}(\cdot\mid y')$,
there exist symbols $x^+,x^-\in\mathcal{X}_{m^\star}$ such that
$h_{m^\star}(x^+)>0$ and $h_{m^\star}(x^-)<0$,
where $h_{m^\star}(x):=\log(p_{m^\star}(x\mid y')/p_{m^\star}(x\mid y))$.
Construct two observation patterns that agree everywhere except at $(m^\star,1)$:
one using $x^+$ and the other using $x^-$.
The score differences under these patterns satisfy $v_+-v_-=h_{m^\star}(x^+)-h_{m^\star}(x^-)>0$,
so at least one of $v_+>0$ or $v_-<0$ holds.
In the first case, the pattern with $x^+$ leads to $\widehat{Y}_r\neq y$ and has conditional
probability at least $\eta_{\min}^{n_{y,y'}(r)}\ge \varepsilon_0$ under $Y=y$.
In the second case, the pattern with $x^-$ leads to $\widehat{Y}_r\neq y'$ with probability
at least $\varepsilon_0$ under $Y=y'$.
\end{proof}

Combining the pairwise tightness lower bound with the minimum sample-size guarantee yields a bound on the surrogate slack.

\begin{lemma}[Surrogate slack for true-feasible plans]
\label{lem:app_true_to_surrogate}
Fix $\alpha$ with $\alpha_{\min}<\varepsilon_0$ and let $r\in\mathcal{R}(\alpha)$.
Then for every $y\in\mathcal{Y}$,
\begin{equation}
\label{eq:app_surrogate_slack}
\overline{P}_e(y;r)
\le
\frac{L-1}{\gamma}\,\alpha_y\,
\sqrt{\sum_{m=1}^K r_m}.
\end{equation}
\end{lemma}

\begin{proof}[Proof of \cref{lem:app_true_to_surrogate}]
Fix $y'\neq y$.
Since $\alpha_{\min}<\varepsilon_0$, \cref{lem:app_min_effective_sample} gives $n_{y,y'}(r)\ge n_0$.
By \cref{lem:app_pairwise_tightness},
$\alpha_y\ge P_e(y;r)\ge \Pr(E^{>}_{y,y'}(r)\mid Y=y)\ge (\gamma/\sqrt{n_{y,y'}(r)})\,\overline{P}_{y,y'}(r)$,
so $\overline{P}_{y,y'}(r)\le (\alpha_y/\gamma)\sqrt{n_{y,y'}(r)}\le (\alpha_y/\gamma)\sqrt{\sum_m r_m}$.
Summing over $y'\neq y$ yields \eqref{eq:app_surrogate_slack}.
\end{proof}

We now assemble the preceding lemmas to bound the additive gap between the surrogate and true optimal costs.

\begin{proposition}[Additive optimality gap]
\label{prop:app_additive_gap}
There exist constants $\alpha_{\min} \in(0,1)$ and $K_0>0$ such that
for sufficiently small $\alpha_{\min}$ and $\mathrm{OPT}(\alpha)<\infty$,
\[
0\le \overline{\mathrm{OPT}}(\alpha)-\mathrm{OPT}(\alpha)
\le K_0\,\log\log\!\left(\frac{1}{\alpha_{\min}}\right).
\]
\end{proposition}

\begin{proof}[Proof of \cref{prop:app_additive_gap}]
Let $r^\star$ be an optimal solution of the true problem with $C(r^\star)=\mathrm{OPT}(\alpha)$
and write $N^\star:=\sum_m r_m^\star$.
We begin by bounding $N^\star$ from above.
Under pairwise identifiability,
$\rho:=\max_{y\neq y'}\min_{s\in[0,1]}\prod_m M_m^{(y,y')}(s)\in(0,1)$.
The uniform design $r^{\mathrm{unif}}(n):=(n,\ldots,n)$ satisfies
$\overline{P}_e(y;r^{\mathrm{unif}}(n))\le (L-1)(\pi_{\max}/\pi_{\min})\rho^n$,
so choosing $n=O(\log(1/\alpha_{\min}))$ makes $r^{\mathrm{unif}}(n)$ feasible.
It follows that
$\mathrm{OPT}(\alpha)=O(\log(1/\alpha_{\min}))$ and $N^\star=O(\log(1/\alpha_{\min}))$.

Although $r^\star$ is feasible for the true problem, it need not be surrogate-feasible.
\Cref{lem:app_true_to_surrogate} quantifies the gap: $\overline{P}_e(y;r^\star)\le (L-1)\alpha_y\sqrt{N^\star}/\gamma$ for every $y$.
To close this gap, consider the padded design $r^\star{}^{+n}=r^\star+n\mathbf{1}$ for $n\in\mathbb{Z}_{\ge 0}$.
For each pair $(y,y')$, the minimizer $s^\star_{y,y'}(r^\star)\in[\delta,1-\delta]$ by \cref{lem:app_centering}.
Since $\overline{P}_{y,y'}(r^\star{}^{+n})\le \overline{P}_{y,y'}(r^\star{}^{+n};s^\star_{y,y'}(r^\star))$
and \cref{lem:app_theta} gives
$\overline{P}_{y,y'}(r^\star{}^{+n};s^\star_{y,y'}(r^\star))\le \vartheta^n\,\overline{P}_{y,y'}(r^\star)$,
summing over $y'\neq y$ yields
$\overline{P}_e(y;r^\star{}^{+n})\le \vartheta^n\,\overline{P}_e(y;r^\star)$.
Choosing
$n:=\lceil\log((L-1)\sqrt{N^\star}/\gamma)/(-\log\vartheta)\rceil$
therefore ensures $\overline{P}_e(y;r^\star{}^{+n})\le \alpha_y$ for all $y$,
so $r^\star{}^{+n}\in\overline{\mathcal{R}}(\alpha)$ and
$\overline{\mathrm{OPT}}(\alpha)\le \mathrm{OPT}(\alpha)+n\,c_{\Sigma}$.

It remains to observe that $n=O(\log\log(1/\alpha_{\min}))$.
Since $N^\star=O(\log(1/\alpha_{\min}))$, we have
$\log(\sqrt{N^\star})=O(\log\log(1/\alpha_{\min}))$, and
the claimed bound follows after absorbing constants into $K_0$.
\end{proof}

We are now ready to prove \cref{thm:optight}. For the reader's convenience, we restate it here.

\medskip
\noindent\textbf{Theorem \ref{thm:optight}} (Optimization-level tightness).
\textit{Suppose the uniform boundedness condition is satisfied. Then for all sufficiently small $\alpha_{\min}:=\min_{y\in\mathcal{Y}}\alpha_y$, if $\mathrm{OPT}(\alpha)<\infty$,
\[
1
\le
\frac{\overline{\mathrm{OPT}}(\alpha)}{\mathrm{OPT}(\alpha)}
\le
1
+
O\!\left(
\frac{
\log\log\!\left(1/\alpha_{\min}\right)
}{
\log\!\left(1/\alpha_{\min}\right)
}\right).
\]
In particular, this ratio approaches $1$ as $\alpha_{\min}\downarrow 0$.}

\begin{proof}[Proof of \cref{thm:optight}]
The lower bound $1\le \overline{\mathrm{OPT}}(\alpha)/\mathrm{OPT}(\alpha)$
follows from the inclusion $\overline{\mathcal{R}}(\alpha)\subseteq\mathcal{R}(\alpha)$.
For the upper bound, \cref{prop:app_additive_gap} gives
$\overline{\mathrm{OPT}}(\alpha)\le \mathrm{OPT}(\alpha)+K_0\log\log(1/\alpha_{\min})$.
It remains to show $\mathrm{OPT}(\alpha)\ge K_{\mathrm{low}}\log(1/\alpha_{\min})$
for some constant $K_{\mathrm{low}}>0$.

Let $r^\star$ be an optimal true-feasible design with
$C(r^\star)=\mathrm{OPT}(\alpha)$ and write $N^\star:=\sum_m r_m^\star$.
Fix $y^\star\in\arg\min_y \alpha_y$ and any competitor $y'\neq y^\star$.
By feasibility and \cref{lem:app_pairwise_tightness},
$\alpha_{\min}\ge \Pr(E^{>}_{y^\star,y'}(r^\star)\mid Y=y^\star)
\ge (\gamma/\sqrt{N^\star})\,\overline{P}_{y^\star,y'}(r^\star)$,
hence
\begin{equation}
\label{eq:app_proxy_upper}
\overline{P}_{y^\star,y'}(r^\star)
\le
\frac{\alpha_{\min}}{\gamma}\sqrt{\frac{\mathrm{OPT}(\alpha)}{c_{\min}}}.
\end{equation}

Define
$w_{\max}:=\max_{m,\,y\neq y',\,s\in[0,1]}(-\log M_m^{(y,y')}(s))<\infty$.
Then
$-\log \overline{P}_{y^\star,y'}(r^\star)
\le \log(\pi_{\max}/\pi_{\min})+(w_{\max}/c_{\min})\,\mathrm{OPT}(\alpha)$,
so
\begin{equation}
\label{eq:app_proxy_lower}
\overline{P}_{y^\star,y'}(r^\star)
\ge
\exp\!\big(-\log(\pi_{\max}/\pi_{\min})-(w_{\max}/c_{\min})\,\mathrm{OPT}(\alpha)\big).
\end{equation}

Combining \eqref{eq:app_proxy_upper} and \eqref{eq:app_proxy_lower} and taking logarithms yields
\[
\frac{w_{\max}}{c_{\min}}\,\mathrm{OPT}(\alpha)
\ge
\log\!\left(\frac{1}{\alpha_{\min}}\right)
-\tfrac{1}{2}\log\!\left(\frac{\mathrm{OPT}(\alpha)}{c_{\min}}\right)
-\log\!\left(\frac{\pi_{\max}}{\pi_{\min}}\right)
-\log\!\left(\frac{1}{\gamma}\right).
\]
Since $\mathrm{OPT}(\alpha)=O(\log(1/\alpha_{\min}))$, the right-hand side is
$\log(1/\alpha_{\min})-O(\log\log(1/\alpha_{\min}))$.
For $\alpha_{\min}$ small enough,
$\mathrm{OPT}(\alpha)\ge K_{\mathrm{low}}\log(1/\alpha_{\min})$.
Dividing the additive gap by this lower bound gives
\[
\frac{\overline{\mathrm{OPT}}(\alpha)}{\mathrm{OPT}(\alpha)}
\le
1+\frac{K_0\log\log(1/\alpha_{\min})}{K_{\mathrm{low}}\log(1/\alpha_{\min})},
\]
which is $1+O(\log\log(1/\alpha_{\min})/\log(1/\alpha_{\min}))$ as $\alpha_{\min}\to 0$, yielding the stated bound in \cref{thm:optight}.
\end{proof}

\subsection{Proof of \cref{thm:afptas}}
\label{app:proof_afptas}

This subsection proves the AFPTAS guarantee stated in \cref{thm:afptas}.
We first collect the supporting lemmas that the proof requires.
Throughout, we use the shorthand
$\alpha_{\min}:=\min_{y}\alpha_y$,
$\pi_{\min}:=\min_{y}\pi(y)$,
$\pi_{\max}:=\max_{y}\pi(y)$,
$c_{\min}:=\min_{m} c_m$,
$c_{\Sigma}:=\sum_{m=1}^K c_m$,
and $\mathcal{P}:=\{(y,y')\in\mathcal{Y}^2: y\neq y'\}$.

\subsubsection{Supporting lemmas}

The first lemma provides a feasibility condition and an upper bound on the total number of queries in any optimal surrogate design.

\begin{lemma}[Crude feasible design and query-count bound]
\label{lem:app_crude_feasible}
Define the pairwise contraction
$\rho:=\max_{y\neq y'}\min_{s\in[0,1]}\prod_{m=1}^K M_m^{(y,y')}(s)$.
Then $\rho\in(0,1)$, and the uniform design
$r^{\mathrm{unif}}(n):=(n,\ldots,n)\in\mathbb{Z}_{\ge 0}^K$ is surrogate-feasible
whenever
\begin{equation}
\label{eq:app_n_unif}
n
\;\ge\;
n_{\mathrm{unif}}
:=
\left\lceil
\frac{
\log\!\Big(\frac{(L-1)\,\pi_{\max}}{\pi_{\min}\,\alpha_{\min}}\Big)
}{
-\log \rho
}
\right\rceil.
\end{equation}
In particular, $\overline{\mathrm{OPT}}(\alpha)\le n_{\mathrm{unif}}\,c_{\Sigma}$.

Let $k_{\max}:=\lceil 2\log 2/(-\log\vartheta)\rceil$,
where $\vartheta\in(0,1)$ is the contraction constant from
\cref{lem:app_theta}.
Every optimal surrogate design $r^\star$ satisfies
\begin{equation}
\label{eq:app_Nmax}
N(r^\star):=\sum_{m=1}^K r_m^\star
\;\le\;
\frac{n_{\mathrm{unif}}\,c_{\Sigma}}{c_{\min}}+k_{\max}\,K
\;=:\;N_{\max}.
\end{equation}
\end{lemma}

\begin{proof}[Proof of \cref{lem:app_crude_feasible}]
For any $y\neq y'$ and $r=r^{\mathrm{unif}}(n)$,
\[
\overline{P}_{y,y'}(r)
=
\min_{s\in[0,1]}
\left(\frac{\pi(y')}{\pi(y)}\right)^{s}
\prod_{m=1}^K \left(M_m^{(y,y')}(s)\right)^{n}
\le
\frac{\pi_{\max}}{\pi_{\min}}\,\rho^n.
\]
Summing over $y'\neq y$ gives
$\overline{P}_e(y;r)\le (L-1)(\pi_{\max}/\pi_{\min})\rho^n$.
This is at most $\alpha_y$ for all $y$ whenever $n\ge n_{\mathrm{unif}}$.
The cost bound $\overline{\mathrm{OPT}}(\alpha)\le n_{\mathrm{unif}}\,c_{\Sigma}$ follows.
For the query-count bound, if $r^\star$ is optimal then
$C(r^\star)\le n_{\mathrm{unif}}\,c_{\Sigma}$ and
$N(r^\star)\le C(r^\star)/c_{\min}\le n_{\mathrm{unif}}\,c_{\Sigma}/c_{\min}$.
Adding the term $k_{\max}K$ accounts for the augmentation used in the
proof of \cref{thm:afptas}.
\end{proof}

The next lemma shows that the pairwise Chernoff bound is Lipschitz in the tilting parameter $s$, with a Lipschitz constant controlled by the query-count bound $N_{\max}$.

\begin{lemma}[Lipschitz bound in $s$]
\label{lem:app_lipschitz_s}
For any ordered pair $p=(y,y')$, any $s,s'\in[0,1]$, and any $r$ with $N(r)\le N_{\max}$,
\begin{equation}
\label{eq:app_pairwise_lip}
\overline{P}_{y,y'}(r;s')
\le
\exp\!\Big(\Lambda\,|s-s'|\Big)\,
\overline{P}_{y,y'}(r;s),
\qquad
\Lambda:=\log\!\Big(\frac{\pi_{\max}}{\pi_{\min}}\Big)+B\,N_{\max}.
\end{equation}
\end{lemma}

\begin{proof}[Proof of \cref{lem:app_lipschitz_s}]
We have $\log \overline{P}_{y,y'}(r;s) = s\log(\pi(y')/\pi(y)) + \sum_{m=1}^K r_m \log M_m^{(y,y')}(s)$.
The first term is Lipschitz in $s$ with constant $|\log(\pi(y')/\pi(y))|\le \log(\pi_{\max}/\pi_{\min})$.
For the second term, differentiating $\log M_m^{(y,y')}(s)$ yields
$(d/ds)\log M_m^{(y,y')}(s) = \E_{q_m^{(y,y')}(\cdot;s)}[\log(p_m(X\mid y')/p_m(X\mid y))]$,
which is bounded in absolute value by $B$ (from the uniform boundedness condition).
By the mean value theorem,
$|\log M_m^{(y,y')}(s) - \log M_m^{(y,y')}(s')| \le B|s-s'|$.
Combining and using $\sum_m r_m = N(r) \le N_{\max}$ gives
$|\log \overline{P}_{y,y'}(r;s) - \log \overline{P}_{y,y'}(r;s')| \le \Lambda|s-s'|$,
which implies \eqref{eq:app_pairwise_lip}.
\end{proof}

The next lemma shows that rounding the discrimination weights downward is conservative and inflates the error bounds by a controlled factor.

\begin{lemma}[Rounding bounds]
\label{lem:app_rounding_bounds}
Fix $s\in[0,1]^{|\mathcal{P}|}$ and let
$\Delta:=\log(1+\varepsilon)/N_{\max}$.
For each model $m$ and pair $p=(y,y')\in\mathcal{P}$, define
$w_{m,p}(s):=-\log M_m^{(y,y')}(s_p)$ and
$\widetilde{w}_{m,p}(s):=\lfloor w_{m,p}(s)/\Delta\rfloor$.
For any $r$ with $N(r)\le N_{\max}$ and any $p\in\mathcal{P}$,
writing $W_p(r;s):=\sum_m r_m w_{m,p}(s)$ and
$\widetilde{W}_p(r;s):=\sum_m r_m \widetilde{w}_{m,p}(s)$, we have
\[
\Delta\,\widetilde{W}_p(r;s)
\;\le\;
W_p(r;s)
\;<\;
\Delta\,\widetilde{W}_p(r;s)+\log(1+\varepsilon).
\]
Consequently,
\[
A_p(s_p)\exp\!\big(-W_p(r;s)\big)
\;\le\;
A_p(s_p)\exp\!\big(-\Delta\,\widetilde{W}_p(r;s)\big)
\;\le\;
(1+\varepsilon)\,A_p(s_p)\exp\!\big(-W_p(r;s)\big),
\]
where $A_p(s_p):=(\pi(y')/\pi(y))^{s_p}$.
\end{lemma}

\begin{proof}[Proof of \cref{lem:app_rounding_bounds}]
The floor definition gives $\Delta\widetilde{w}_{m,p}\le w_{m,p} < \Delta(\widetilde{w}_{m,p}+1)$.
Multiplying by $r_m$ and summing over $m$ yields
$\Delta\,\widetilde{W}_p \le W_p < \Delta\,\widetilde{W}_p + \Delta\,N(r) \le \Delta\,\widetilde{W}_p + \log(1+\varepsilon)$.
The two-sided bound on $A_p\exp(-\cdot)$ follows by exponentiating and using monotonicity.
\end{proof}

The next lemma establishes the DP recursion that Algorithm~\ref{alg:afptas} uses to solve each fixed-$s$ subproblem.
Fix $s\in[0,1]^{|\mathcal{P}|}$ and let $T_{\max}:=\lceil B N_{\max}/\Delta\rceil$.
For each DP state $t=(t_p)_{p\in\mathcal{P}}\in\{0,\ldots,T_{\max}\}^{|\mathcal{P}|}$, define the conservative error bound
\begin{equation}
\label{eq:app_dp_error}
\widehat{P}_e(y;t;s)
:=
\sum_{y'\neq y}
A_{(y,y')}\!\big(s_{(y,y')}\big)\,
\exp\!\big(-\Delta\, t_{(y,y')}\big),
\qquad y\in\mathcal{Y}.
\end{equation}

\begin{lemma}[DP recursion]
\label{lem:app_dp_recursion}
Define the DP value function
\[
\mathrm{DP}[t]
:=
\min\Big\{
C(r): r\in\mathbb{Z}_{\ge 0}^K,\ \widetilde{W}_p(r;s)\ge t_p\ \ \forall p\in\mathcal{P}
\Big\},
\]
with the convention $\mathrm{DP}[t]=+\infty$ if the set is empty.
Then $\mathrm{DP}[\mathbf{0}]=0$, and for $t\neq \mathbf{0}$,
\begin{equation}
\label{eq:app_dp_recursion}
\mathrm{DP}[t]
=
\min_{m\in\{1,\ldots,K\}}
\Big\{c_m+\mathrm{DP}\big[(t-\widetilde{w}_m)_+\big]\Big\},
\end{equation}
where $(t-\widetilde{w}_m)_+$ denotes componentwise
$\max\{t_p-\widetilde{w}_{m,p}(s),\,0\}$.
Moreover, the minimizing $m$ in \eqref{eq:app_dp_recursion} provides a
backpointer that recovers an optimal design $r$ for $\mathrm{DP}[t]$.
\end{lemma}

\begin{proof}[Proof of \cref{lem:app_dp_recursion}]
If $t=\mathbf{0}$, then $r=\mathbf{0}$ is feasible with cost $0$, so $\mathrm{DP}[\mathbf{0}]=0$.
Now let $t\neq \mathbf{0}$ and suppose $\mathrm{DP}[t]<+\infty$.
Let $r$ be an optimizer for $\mathrm{DP}[t]$.
Since $t\neq \mathbf{0}$, we must have $r\neq \mathbf{0}$, so there exists $m$ with $r_m\ge 1$.
Setting $r':=r-e_m$ gives $C(r')=C(r)-c_m$ and
$\widetilde{W}_p(r';s)=\widetilde{W}_p(r;s)-\widetilde{w}_{m,p}(s)\ge (t-\widetilde{w}_m)_{+,p}$
for all $p$.
Therefore $r'$ is feasible for $\mathrm{DP}[(t-\widetilde{w}_m)_+]$, giving
$C(r)\ge c_m+\mathrm{DP}[(t-\widetilde{w}_m)_+]$.
Minimizing over $m$ proves the ``$\ge$'' direction.

For the reverse direction, fix any $m$ with $\mathrm{DP}[(t-\widetilde{w}_m)_+]<+\infty$ and let $r'$ be an optimizer for that subproblem.
Setting $r:=r'+e_m$ gives $\widetilde{W}_p(r;s)=\widetilde{W}_p(r';s)+\widetilde{w}_{m,p}(s)\ge t_p$ for all $p$, so $r$ is feasible for $\mathrm{DP}[t]$ with cost $c_m+\mathrm{DP}[(t-\widetilde{w}_m)_+]$.
Taking the minimum over $m$ yields the ``$\le$'' direction.
\end{proof}

The following proposition connects the DP output to surrogate-feasibility.

\begin{proposition}[Fixed-$s$ DP feasibility]
\label{prop:app_fixeds_dp}
Fix $s\in[0,1]^{|\mathcal{P}|}$.
Let $t^\dagger$ be any minimizer of $\mathrm{DP}[t]$ over the states
$t\in\{0,\ldots,T_{\max}\}^{|\mathcal{P}|}$ satisfying
$\widehat{P}_e(y;t;s)\le \alpha_y$ for all $y\in\mathcal{Y}$.
Let $r^\dagger$ be the design recovered from the backpointers at $t^\dagger$.
Then $r^\dagger$ is surrogate-feasible for the original problem \eqref{eq:surrogate_problem}.

Moreover, if some design $r'$ with $N(r')\le N_{\max}$ satisfies
$\widehat{P}_e(y;\widetilde{W}(r';\,s);s)\le \alpha_y$ for all $y$,
then
$\mathrm{DP}[\widetilde{W}(r';s)]\le C(r')$
and consequently
$C(r^\dagger)\le C(r')$.
\end{proposition}

\begin{proof}[Proof of \cref{prop:app_fixeds_dp}]
By \cref{lem:app_dp_recursion}, the recovered design $r^\dagger$ satisfies
$\widetilde{W}_p(r^\dagger;s)\ge t^\dagger_p$ for all $p$.
\Cref{lem:app_rounding_bounds} then gives
$W_p(r^\dagger;s)\ge \Delta\, t^\dagger_p$, so
\[
\begin{aligned}
    \overline{P}_e(y;r^\dagger;s)
&=
\sum_{y'\neq y}A_{(y,y')}(s_{(y,y')})\exp(-W_{(y,y')}(r^\dagger;s))
\le
\sum_{y'\neq y}A_{(y,y')}(s_{(y,y')})\exp(-\Delta\, t^\dagger_{(y,y')})\\
&=
\widehat{P}_e(y;t^\dagger;s)
\le \alpha_y.
\end{aligned}
\]
Surrogate-feasibility of $r^\dagger$ for the original problem \eqref{eq:surrogate_problem}
then follows from \cref{prop:joint_reformulation}, since feasibility for any fixed $s$ implies
$\overline{P}_e(y;r^\dagger)\le \overline{P}_e(y;r^\dagger;s)\le \alpha_y$ for all $y$.

For the cost comparison, the design $r'$ is feasible for
$\mathrm{DP}[\widetilde{W}(r';s)]$ by definition, so
$\mathrm{DP}[\widetilde{W}(r';s)]\le C(r')$.
Since $t^\dagger$ minimizes $\mathrm{DP}[t]$ over all feasible states,
$C(r^\dagger)=\mathrm{DP}[t^\dagger]\le \mathrm{DP}[\widetilde{W}(r';s)]\le C(r')$.
\end{proof}

The final supporting lemma provides a universal lower bound on the optimal surrogate cost.

\begin{lemma}[Cost lower bound]
\label{lem:app_cost_lb}
Any surrogate-feasible design $r\in\overline{\mathcal{R}}(\alpha)$ satisfies
\begin{equation}
\label{eq:app_cost_lb}
C(r)
\;\ge\;
\frac{c_{\min}}{B}\,
\log\!\Big(\frac{(L-1)\,\pi_{\min}}{\pi_{\max}\,\alpha_{\min}}\Big)
\qquad
\text{whenever }\ \alpha_{\min}<\frac{(L-1)\,\pi_{\min}}{\pi_{\max}}.
\end{equation}
\end{lemma}

\begin{proof}[Proof of \cref{lem:app_cost_lb}]
For any pair $(y,y')$ and any $s\in[0,1]$,
$M_m^{(y,y')}(s)\ge e^{-B}$ by the uniform boundedness condition, so
$\prod_m M_m^{(y,y')}(s)^{r_m}\ge e^{-B N(r)}$.
Also, $(\pi(y')/\pi(y))^s\ge \pi_{\min}/\pi_{\max}$ for $s\in[0,1]$.
Therefore
$\overline{P}_{y,y'}(r)\ge (\pi_{\min}/\pi_{\max})\,e^{-B\,N(r)}$.
Picking $y$ with $\alpha_y=\alpha_{\min}$ and summing over $y'\neq y$ gives
$\alpha_{\min}\ge \overline{P}_e(y;r) \ge (L-1)(\pi_{\min}/\pi_{\max})\,e^{-B\,N(r)}$.
Rearranging yields $N(r)\ge \frac{1}{B}\log((L-1)\pi_{\min}/(\pi_{\max}\alpha_{\min}))$,
and $C(r)\ge c_{\min}\,N(r)$.
\end{proof}

\subsubsection{Proof of \cref{thm:afptas}}

For the reader's convenience, we restate the theorem.

\medskip
\noindent\textbf{\cref{thm:afptas}} (AFPTAS guarantee).
\textit{The design $\widehat{r}$ returned by Algorithm~\ref{alg:afptas} is surrogate-feasible, and for sufficiently small $\alpha_{\min}$, the
multiplicative guarantee
$C(\widehat{r})\le (1+\varepsilon)\,\overline{\mathrm{OPT}}(\alpha)$
holds. The runtime is polynomial in $K$, $\log(1/\alpha_{\min})$, and $1/\varepsilon$ for fixed $L$.}

\begin{proof}[Proof of \cref{thm:afptas}]

\noindent\textbf{Surrogate-feasibility.}
Every candidate design $r^\dagger(s)$ returned by the fixed-$s$ DP subroutine satisfies $\overline{P}_e(y;r^\dagger(s);s)\le \alpha_y$ for all $y\in\mathcal{Y}$
by \cref{prop:app_fixeds_dp}.
\Cref{prop:joint_reformulation} then implies
$\overline{P}_e(y;r^\dagger(s))\le \overline{P}_e(y;r^\dagger(s);s)\le \alpha_y$ for all $y$, so each candidate is surrogate-feasible.
Since $\widehat{r}$ is the cheapest among all candidates, it is surrogate-feasible as well.

\noindent\textbf{Runtime.}
The DP at each grid point has $(T_{\max}+1)^{|\mathcal{P}|}$ states, and each recursion step in \eqref{eq:app_dp_recursion} requires $O(K)$ work.
Since $T_{\max}=O(BN_{\max}^2/\log(1+\varepsilon))$ and $|\mathcal{P}|=L(L-1)$ is fixed,
the per-grid-point cost is polynomial in $K$, $N_{\max}$, and $1/\varepsilon$.
The grid $\mathcal{G}_\varepsilon$ contains $O(\Lambda/\log(1+\varepsilon))^{|\mathcal{P}|}$ points, which is also polynomial in $N_{\max}$ and $1/\varepsilon$ for fixed $L$.
Since $N_{\max}=O(\log(1/\alpha_{\min}))$ by \cref{lem:app_crude_feasible}, the overall runtime is polynomial in $K$, $\log(1/\alpha_{\min})$, and $1/\varepsilon$.

\noindent\textbf{Multiplicative cost bound.}
It remains to prove $C(\widehat{r})\le (1+\varepsilon)\,\overline{\mathrm{OPT}}(\alpha)$ for sufficiently small $\alpha_{\min}$.
Let $r^\star$ denote an optimal surrogate design with cost $C(r^\star)=\overline{\mathrm{OPT}}(\alpha)$.
The proof proceeds by constructing an explicit feasible solution
for the DP at some grid point $\bar{s}\in\mathcal{G}_\varepsilon^{|\mathcal{P}|}$
whose cost we can control, then arguing that the DP output is at
least as cheap as this constructed solution.
To absorb the two sources of approximation error that the algorithm introduces
(grid discretization and weight rounding), we first augment $r^\star$ with
additional queries to every model, creating a $(1+\varepsilon)^2$ margin of
slack in the surrogate constraints.

\noindent
\textbf{Constructing a near-feasible augmented design.}
For each ordered pair $p=(y,y')\in\mathcal{P}$, select any minimizer
\[
s_p^\star\in\arg\min_{s\in[0,1]}\overline{P}_{y,y'}(r^\star;s),
\]
which exists because $[0,1]$ is compact and
$s\mapsto\overline{P}_{y,y'}(r^\star;s)$ is continuous.
Setting $s^\star:=(s_p^\star)_{p\in\mathcal{P}}\in[0,1]^{|\mathcal{P}|}$
and applying \cref{prop:joint_reformulation} gives
\[
\overline{P}_e(y;r^\star;s^\star)
=
\sum_{y'\neq y}\overline{P}_{y,y'}(r^\star;s^\star_{(y,y')})
=
\sum_{y'\neq y}\min_{s\in[0,1]}\overline{P}_{y,y'}(r^\star;s)
=
\overline{P}_e(y;r^\star)
\le \alpha_y
\]
for every $y\in\mathcal{Y}$,
so the pair $(r^\star,s^\star)$ is feasible for the joint formulation
\eqref{eq:joint_reformulation}.

We now verify that every component of $s^\star$ lies in the interior interval $[\delta,1-\delta]$.
Since $r^\star$ is surrogate-feasible, \cref{thm:statewise_chernoff} gives
\[
P_e(y;r^\star)\le \overline{P}_e(y;r^\star)\le \alpha_y
\qquad \forall\, y\in\mathcal{Y},
\]
so $r^\star$ is also feasible for the original robust problem
\eqref{eq:opt_problem}.
The assumption $\alpha_{\min}\le \varepsilon_0$ then places us
in the regime of \cref{lem:app_min_effective_sample}, which yields
$n_{y,y'}(r^\star)\ge n_0$ for all $y\neq y'$.
By construction $n_0\ge n_{\mathrm{cen}}$, so
$n_{y,y'}(r^\star)\ge n_{\mathrm{cen}}$ and \cref{lem:app_centering}
gives $s^\star_p\in[\delta,1-\delta]$ for every $p\in\mathcal{P}$.

Define $k(\varepsilon):=\lceil 2\log(1+\varepsilon)/(-\log\vartheta)\rceil$
and the augmented design
$\widetilde{r}:=r^\star + k(\varepsilon)\,\mathbf{1}$,
where $\mathbf{1}=(1,\ldots,1)\in\mathbb{Z}_{\ge 0}^K$.
Because $s^\star_p\in[\delta,1-\delta]$ for every $p$, the definition
of $\vartheta$ in \cref{lem:app_theta} gives
$\prod_{m=1}^K M_m^{(y,y')}(s^\star_p)\le \vartheta$ for every
ordered pair $p=(y,y')$.
For any such pair, the pairwise Chernoff bound of the augmented design
evaluated at $s_p^\star$ satisfies
\begin{align}
\overline{P}_{y,y'}\!\big(\widetilde{r};s^\star_p\big)
&=
\left(\frac{\pi(y')}{\pi(y)}\right)^{s^\star_p}
\prod_{m=1}^K \left(M_m^{(y,y')}(s^\star_p)\right)^{r^\star_m + k(\varepsilon)}
\nonumber\\
&=
\left(\frac{\pi(y')}{\pi(y)}\right)^{s^\star_p}
\prod_{m=1}^K \left(M_m^{(y,y')}(s^\star_p)\right)^{r^\star_m}
\cdot
\prod_{m=1}^K \left(M_m^{(y,y')}(s^\star_p)\right)^{k(\varepsilon)}
\nonumber\\
&=
\overline{P}_{y,y'}\!\big(r^\star;s^\star_p\big)
\cdot
\left(\prod_{m=1}^K M_m^{(y,y')}(s^\star_p)\right)^{k(\varepsilon)}
\nonumber\\
&\le
\overline{P}_{y,y'}\!\big(r^\star;s^\star_p\big)\,\vartheta^{k(\varepsilon)},
\label{eq:app_pairwise_augment}
\end{align}
where the final inequality uses
$\prod_{m=1}^K M_m^{(y,y')}(s^\star_p)\le \vartheta$.
The definition of $k(\varepsilon)$ guarantees that
$k(\varepsilon)\ge 2\log(1+\varepsilon)/(-\log\vartheta)$,
which implies
\[
\vartheta^{k(\varepsilon)}
=
e^{\,k(\varepsilon)\log\vartheta}
\le
e^{\,-2\log(1+\varepsilon)}
=
\frac{1}{(1+\varepsilon)^2}.
\]
Substituting this into \eqref{eq:app_pairwise_augment} gives
$\overline{P}_{y,y'}(\widetilde{r};s^\star_p)
\le
\overline{P}_{y,y'}(r^\star;s^\star_p)/(1+\varepsilon)^2$.
Summing over $y'\neq y$ then yields
\begin{equation}
\label{eq:app_rtilde_strong}
\overline{P}_e(y;\widetilde{r};s^\star)
=
\sum_{y'\neq y}
\overline{P}_{y,y'}\!\big(\widetilde{r};s^\star_{(y,y')}\big)
\le
\frac{1}{(1+\varepsilon)^2}
\sum_{y'\neq y}
\overline{P}_{y,y'}\!\big(r^\star;s^\star_{(y,y')}\big)
=
\frac{\overline{P}_e(y;r^\star;s^\star)}{(1+\varepsilon)^2}
\le
\frac{\alpha_y}{(1+\varepsilon)^2}
\end{equation}
for every $y\in\mathcal{Y}$.
The augmented design therefore satisfies the surrogate
constraints with a quadratic margin of $(1+\varepsilon)^2$.
This margin will absorb the two approximation losses that follow.

We also verify that $N(\widetilde{r})\le N_{\max}$.
\Cref{lem:app_crude_feasible} gives
$C(r^\star)\le n_{\mathrm{unif}}\,c_{\Sigma}$, and since
$C(r^\star)=\sum_m c_m r_m^\star\ge c_{\min}\sum_m r_m^\star
=c_{\min}\,N(r^\star)$, we obtain
$N(r^\star)\le n_{\mathrm{unif}}\,c_{\Sigma}/c_{\min}$.
Because $\log(1+\varepsilon)\le \log 2$ for $\varepsilon\in(0,1]$, the
ceiling definition gives $k(\varepsilon)\le k_{\max}$, so
\[
N(\widetilde{r})
=
N(r^\star)+k(\varepsilon)\,K
\le
\frac{n_{\mathrm{unif}}\,c_{\Sigma}}{c_{\min}}+k_{\max}\,K
=
N_{\max}.
\]
This confirms that \cref{lem:app_lipschitz_s,lem:app_rounding_bounds}
apply to $\widetilde{r}$.

\noindent
\textbf{Snapping $s^\star$ to the grid.}
The grid $\mathcal{G}_\varepsilon$ in Algorithm~\ref{alg:afptas} has mesh
$h=\log(1+\varepsilon)/\Lambda$.
There exists a grid point
$\bar{s}\in\mathcal{G}_\varepsilon^{|\mathcal{P}|}$ such that
$|s^\star_p-\bar{s}_p|\le h$ for every $p\in\mathcal{P}$.
Applying \cref{lem:app_lipschitz_s} to each pair $(y,y')$ gives
\[
\begin{aligned}
    \overline{P}_{y,y'}(\widetilde{r};\bar{s}_{(y,y')})
&\le
\exp\!\big(\Lambda\,|s^\star_{(y,y')}-\bar{s}_{(y,y')}|\big)\,
\overline{P}_{y,y'}(\widetilde{r};s^\star_{(y,y')})
\le
\exp(\Lambda\, h)\,
\overline{P}_{y,y'}(\widetilde{r};s^\star_{(y,y')})
\\
&= (1+\varepsilon)\,\overline{P}_{y,y'}(\widetilde{r};s^\star_{(y,y')}),
\end{aligned}
\]
where the final equality uses $\Lambda\, h = \log(1+\varepsilon)$.
Summing over $y'\neq y$ and applying \eqref{eq:app_rtilde_strong} gives
\begin{equation}
\label{eq:app_rtilde_grid}
\overline{P}_e(y;\widetilde{r};\bar{s})
=
\sum_{y'\neq y}
\overline{P}_{y,y'}\!\big(\widetilde{r};\bar{s}_{(y,y')}\big)
\le
(1+\varepsilon)
\sum_{y'\neq y}
\overline{P}_{y,y'}\!\big(\widetilde{r};s^\star_{(y,y')}\big)
=
(1+\varepsilon)\,\overline{P}_e(y;\widetilde{r};s^\star)
\le
\frac{\alpha_y}{1+\varepsilon}
\end{equation}
for every $y\in\mathcal{Y}$.
One of the two factors of $(1+\varepsilon)$ in the margin has now been consumed.

\noindent
\textbf{Rounding weights to DP states.}
Let $t:=\widetilde{W}(\widetilde{r};\bar{s})$ denote the rounded weight
vector computed by Algorithm~\ref{alg:afptas} at grid point $\bar{s}$.
By \cref{lem:app_rounding_bounds}, rounding each weight
$w_{m,p}(\bar{s}_p)$ down to its
nearest multiple of $\Delta=\log(1+\varepsilon)/N_{\max}$ introduces
a per-model error of at most $\Delta$, and summing over the
$N(\widetilde{r})\le N_{\max}$ queries yields a total rounding gap of
at most $N_{\max}\cdot\Delta=\log(1+\varepsilon)$ in each coordinate.
Exponentiating, this gives
\[
\widehat{P}_e(y;t;\bar{s})
\le
\exp\!\big(\log(1+\varepsilon)\big)\,\overline{P}_e(y;\widetilde{r};\bar{s})
=
(1+\varepsilon)\,\overline{P}_e(y;\widetilde{r};\bar{s}).
\]
Combining with \eqref{eq:app_rtilde_grid} yields
\[
\widehat{P}_e(y;t;\bar{s})
\le
(1+\varepsilon)\cdot\frac{\alpha_y}{1+\varepsilon}
=
\alpha_y,
\qquad \forall\, y\in\mathcal{Y}.
\]
The second and final factor of $(1+\varepsilon)$ in the margin has been
consumed, and $\widetilde{r}$ is feasible for the DP feasibility
condition at the grid point $\bar{s}$.

\noindent
\textbf{Comparing costs.}
Because $\widetilde{r}$ satisfies the DP feasibility condition at $\bar{s}$,
\cref{prop:app_fixeds_dp} gives
$C(r^\dagger(\bar{s}))\le C(\widetilde{r})$.
Algorithm~\ref{alg:afptas} outputs the cheapest design across all grid points, so
\begin{equation}
\label{eq:app_cost_chain}
C(\widehat{r})
\le
C(r^\dagger(\bar{s}))
\le
C(\widetilde{r})
=
C(r^\star)+k(\varepsilon)\,c_{\Sigma}
=
\overline{\mathrm{OPT}}(\alpha)+k(\varepsilon)\,c_{\Sigma}.
\end{equation}

\noindent
\textbf{Converting the additive bound to a multiplicative bound.}
It remains to show that $k(\varepsilon)\,c_{\Sigma}\le \varepsilon\,\overline{\mathrm{OPT}}(\alpha)$
for sufficiently small $\alpha_{\min}$.
Define $\alpha_0(\varepsilon):=\frac{(L-1)\pi_{\min}}{\pi_{\max}}\exp\!\big(-\frac{B\,k(\varepsilon)\,c_{\Sigma}}{\varepsilon\,c_{\min}}\big)$
and assume $\alpha_{\min}\le \alpha_0(\varepsilon)$.
Taking logarithms of both sides of
$\alpha_{\min}\le \alpha_0(\varepsilon)$ gives
\[
\log\alpha_{\min}
\le
\log\!\Big(\frac{(L-1)\,\pi_{\min}}{\pi_{\max}}\Big)
-\frac{B\,k(\varepsilon)\,c_{\Sigma}}{\varepsilon\,c_{\min}},
\]
which rearranges to
\begin{equation}
\label{eq:app_threshold_equiv}
\frac{B\,k(\varepsilon)\,c_{\Sigma}}{\varepsilon\,c_{\min}}
\le
\log\!\Big(\frac{(L-1)\,\pi_{\min}}{\pi_{\max}\,\alpha_{\min}}\Big).
\end{equation}
Applying \cref{lem:app_cost_lb} and
using \eqref{eq:app_threshold_equiv} yields
\[
\overline{\mathrm{OPT}}(\alpha)
\;\ge\;
\frac{c_{\min}}{B}\,
\log\!\Big(\frac{(L-1)\,\pi_{\min}}{\pi_{\max}\,\alpha_{\min}}\Big)
\;\ge\;
\frac{c_{\min}}{B}\cdot
\frac{B\,k(\varepsilon)\,c_{\Sigma}}{\varepsilon\,c_{\min}}
\;=\;
\frac{k(\varepsilon)\,c_{\Sigma}}{\varepsilon}.
\]
Multiplying both sides by $\varepsilon$ gives
$k(\varepsilon)\,c_{\Sigma}\le \varepsilon\,\overline{\mathrm{OPT}}(\alpha)$.
Substituting this into \eqref{eq:app_cost_chain} produces
\[
C(\widehat{r})
\le
\overline{\mathrm{OPT}}(\alpha)+k(\varepsilon)\,c_{\Sigma}
\le
\overline{\mathrm{OPT}}(\alpha)+\varepsilon\,\overline{\mathrm{OPT}}(\alpha)
=
(1+\varepsilon)\,\overline{\mathrm{OPT}}(\alpha),
\]
which is the multiplicative guarantee \eqref{eq:multiplicative}.
\end{proof}

\end{document}